\documentclass[12pt]{article}
\usepackage[utf8]{inputenc}
\usepackage[a4paper, total={6.5in, 9.5in}]{geometry}

\usepackage{amssymb}
\usepackage{amsmath} 
\usepackage{amsfonts}
\usepackage{eurosym}
\usepackage{geometry}
\usepackage{ulem}
\usepackage{graphicx}
\usepackage{caption}
\usepackage{color}
\usepackage{setspace}
\usepackage{comment}
\usepackage{footmisc}
\usepackage{natbib}
\usepackage{pdflscape}
\usepackage{subfigure}
\usepackage{array}
\usepackage{amsthm} % to set up the proof format
\usepackage{cleveref} % to ref equations \Cref{}
\usepackage{mathrsfs} 
\usepackage{verbatim}
\usepackage{mathtools}
\usepackage{tikz}
\usepackage{verbatim}

\newtheorem{theorem}{Theorem}
\newtheorem{assumption}{Assumption}

\newtheorem{lemma}{Lemma}
\newtheorem{proposition}{Proposition}
\newtheorem{corollary}{Corollary}

\newtheorem{definition}{Definition}

\newtheorem*{observation}{Observation}

\begin{document}

\begin{titlepage}
\title{Private Information Acquisition and Preemption: a Strategic Wald Problem
%\thanks{abc}
}
\author{Guo Bai
\thanks{
PhD student, Department of Economics, University College London.
Address: Drayton House, 30 Gordon Street, London, WC1H 0AX.
Email: guo.bai.15@ucl.ac.uk.
I am grateful to Martin Cripps, Konrad Mierendorff and Deniz Kattwinkel for supervision and support. 
} 
%\and 
%John Smith\thanks{abc}
}
\date{\today}

\maketitle

\begin{abstract}
\noindent 
This paper studies a dynamic information acquisition model with payoff externalities.
Two players can acquire costly information about an unknown state before taking a safe or risky action. 
Both information and the action taken are private.
The first player to take the risky action has an advantage but whether the risky action is profitable depends on the state.
The players face the tradeoff between being first and being right.
In equilibrium, for different priors, there exist three kinds of randomisation: when the players are pessimistic, they enter the competition randomly; 
when the players are less pessimistic, they acquire information and then randomly stop;
when the players are relatively optimistic, they randomly take an action without acquiring information.
\\
\vspace{0in}\\
%\noindent\textbf{Keywords:} 
\\
%\vspace{0in}\\
%\noindent\textbf{JEL Codes:} 
\bigskip
\end{abstract}
\setcounter{page}{0}
\thispagestyle{empty}
\end{titlepage}
\pagebreak 
\newpage

%\begin{abstract}
%abstract
%\end{abstract}

\section{Introduction}
\label{sec:intro}

Research and development (R\&D) of a new technology is often competitive as there sometimes exists a first-mover advantage.
The first company to verify the feasibility of the technology and conduct mass production generates more sales than any late competitor.
A company's R\&D hence affects its competitor's profit and vice versa.
However, 
companies may not be able to see its competitor's breakthrough or their start of mass production.
A company may fall behind without even noticing it.
Similarly, the company may not see its competitor's breakdown or their silent exit from competition.
Then, the company's further R\&D will lead to either a `win' in a doomed-to-fail project or a belated exit.

A similar story also applies to other economic or social activities.
For example, researchers compete for novel results.
It takes time for a paper to become public after the researchers find a profound result. 
Moreover, negative results are often not reported, hence not observed.
It is difficult for the researchers to know other researchers' private findings. 
However, 
when other researchers publish their findings first, the novelty disappears.
Other researchers' early publication will decrease the researchers' credit received from the similar findings. 
By the time the researchers observe the publication, it is already too late, and
the cost associated with the research has been incurred already.
Not being able to see other researchers' progress makes it harder for the researchers to evaluate the potential credit they would get from the project.

Another example is market entry.
If a company is the first to enter a new market, they may capture a higher market share.
But before entry, they need to investigate whether the market is of high or low demand.
This investigation takes time and the preparation of entering a new market is also not \textit{immediately} observable.
The company may decide to enter a new market with the belief that they will be the first, but later realise that their competitor has started the whole process earlier than themselves.
Their competitor will then be the pioneer and the company can only follow.
At the time the company observes the competing company's decision, it is too late to reverse their own decision.
Because of the delay of observing the competitor's action, entering a new market is essentially a private action.

The objective of this paper is to understand the tradeoff between being first and being right when information acquisition takes time and opponent's action taken is not immediately observable.
I focus on understanding the players' equilibrium strategy in the presence of payoff externalities while shutting down the information externalities.
The payoff externalities stem from the first-mover advantage.
One's early success reduces the opponent's payoff from late success.
The absence of information externalities is because of the unobervable actions.
There is no information spillover because they do not get any direct or indirect information from each other.

%\paragraph{Model}
To model the interactions described above, 
this paper studies a dynamic model with private information acquisition and payoff externalities.
Both players can acquire costly information about an unknown payoff-relevant state before taking an irreversible action.
The unknown state is either high or low and is constant over time.
If the state is high, it is optimal for the player to take the risky action, and
if the state is low, it is optimal to take the safe action.
Competition is modelled by the payoff externality.
The first player who stops acquiring information and takes the risky action gets a higher payoff in both states.
For example,
two companies competing in the R\&D of a new self-driving car technology.
The unknown payoff-relevant state is the feasibility of this technology.
The R\&D is abstracted as paying a fixed cost to acquire information about the technology.
The risky action is to start the mass production of the self-driving cars and the safe action is to exit.
The first to start the mass production will launch the car early and consequently gain a greater market share.
Whether a company starts mass production is not immediately observable.

Both players have access to conditionally independent signals from an identical information source.
These generate Poisson breakthroughs and breakdowns.
A breakthrough is good news that reveals the high state and
a breakdown is bad news that reveals the low state.
I also assume that no news is good news.
In the self-driving car example, 
this information structure describes the incremental development with potential breakthroughs and breakdowns.
The company makes incremental improvement, and at the same time, they can come across conclusive evidence that proves or disproves the feasibility of the technology.

Both the information acquisition and the action taken are assumed to be private.
This describes an anonymous information acquisition environment: a player does not see her competitor's signals or actions, but she is aware of the competitor's existence.
This shuts down the information externality where the player learns something from the competitor's action and focuses on the competition-generated payoff externality.

My setup is related to the bandit problem, the classical framework studying R\&D races, but is different in the following aspects.
In the bandit problem with the safe and the risky arms, pulling the safe arm is similar to taking the \textit{safe action} in this model,
while pulling the risky arm is similar to \textit{acquiring information}.
The payoff generated by the risky arm depends on the nature of the risky arm, which is similar as the \textit{state} in this model.
The difference is that in the bandit problem, the player pulling the risky arm gets \textit{information} about the arm from the \textit{payoff} generated by the arm.
A high payoff from the risky arm indicates a good arm.
In other words, the payoff \textit{is} the information.
In contrast, in my model, I separate the information and the payoff from an action.
The breakthroughs, breakdowns, or the lack thereof, only contain \textit{information} about the state. 
If the players want to use the information, they need to take an action.
Acquiring information itself does not give the player any payoff.
Instead, it incurs a positive information cost.
The R\&D is modelled as a costly activity that contains information about the feasibility of the technology, but does not give the company direct return.
The company gets return from the R\&D only if they take actions using the information generated.

This separation of information and the payoffs allows for the discussions not only on the timing to stop the R\&D but also the timing to start the R\&D.
Take the self-driving car competition example again. 
Suppose the feasibility of the R\&D is known, but the profitability is unknown, 
e.g. the demand for self-driving cars is uncertain.
The companies are interested in when or whether to start the R\&D.
Doing R\&D is only profitable if the demand is high.
Early start of the R\&D means early success and higher profit.
Then, companies can pay a cost to do online surveys to better estimate the demand.
Companies doing online surveys is another example that can be modelled as the costly information acquisition.
The risky action in this setting is to start the R\&D on self-driving cars and the safe action is to abandon this self-driving car project.

% results

The main result in this paper shows that in equilibrium, players use random stopping strategy if they acquire information.
This is significantly different from the single DM case.
In the single DM case, the deterministic cutoff strategy is the optimal one: to acquire information if the belief is in the intermediate range and to take an action if the belief is extreme.
In my model with competition and first-mover advantage, for relatively pessimistic priors, 
the players use the random stopping strategy if they acquire information.
That is, they acquire information up to some point and then randomly stop and take the risky action.
This is a result of the interaction of the learning motive and the preemption motive.
At the early stage of information acquisition, players are uncertain about the state and hence have stronger learning motives.
In addition, since they are relatively pessimistic about the state, they also believe that the probability of being preempted is low. 
The value associated with taking the risky action is large if the state happens to be high.
At this point,
acquiring information is a strategic complement.
If the opponent has not stopped and taken the risky action, then, the gain from acquiring information is high and the player is willing to acquire more information.
The more information the opponent is expected to acquire, the more information the player is willing to acquire.
After the player has been acquiring information for some time, she becomes more optimistic about the state (because no news is good news). 
However, the value associated with taking the risky action conditional on the high state becomes lower.
This is because the probability of the opponent taking the risky action increases.
The player then has a stronger preemption motive to stop acquiring information and to preempt the opponent.

When the prior is relatively optimistic, there exists an equilibrium where the players randomise between the safe and risky action immediately without acquiring information.
The existence of this equilibrium is a result of the strong preemption motive and the weak learning motive.
Players undercut the time at which the other takes risky action until there is no room for further preemption.
However, the players are not optimistic enough to take the risky action with probability one.
With a positive probability of taking the safe action, the player secures a zero payoff.
At the same time, the opponent's learning is deterred.
An interesting result is that even when the information cost is zero, this equilibrium still exists.
This is because in my model, the total information cost consists of the exogenous information cost $c$ and an endogenous information cost from being preempted.
The endogenous information cost is determined by the opponent's strategy.
When the exogenous information cost vanishes, the endogenous information cost does not vanish and hence the total information cost is still positive.
This positive total information cost gives rise to this equilibrium where the players take immediate action at time zero without acquiring information.

For sufficiently (but not extremely) pessimistic prior, there exists an asymmetric equilibrium where one player takes the safe action immediately and the opponent acquires information.
Information acquisition here has strategic substitutes' property where one player acquires information only if the opponent does not.
The intuition is when the prior is sufficiently pessimistic, it takes longer for the belief to drift to the random stopping point.
The expected information cost is hence higher.
The value associated with taking the risky action must be high enough to compensate for the higher information cost so that the player is willing to acquire information.
When the opponent drops out, the value associated with the risky action is the highest and when the opponent participates in the competition, this value decreases.
In symmetric equilibrium, the player randomises between these two roles where she mixes between the random stopping strategy and the immediate safe action at time zero.

% equilibrium outcome -- in both states

\paragraph{Related literature}
% only say things that distinguish my model from the others'.

I study \textit{strategic} dynamic information acquisition in an optimal stopping framework introduced by 
\cite{Wald1945, Wald1947}.
The single decision maker's optimal stopping problem has been studied in both drift-diffusion models and Poisson models.
\cite{Fudenberg2018} study the relationship between decision time and accuracy of the action in a modified drift-diffusion model.
\cite{Ke2019} and \cite{Nikandrova2018} investigate how decision maker optimally acquire information about two different alternatives.
The former considers a drift-diffusion model and the latter considers a Poisson model.
\cite{Che2019} and \cite{Mayskaya2020} study how a decision maker optimally choose the bias of the information source in models with Poisson signals.
My model adopts the Poisson signal structure and more importantly, considers the \textit{strategic} interaction between players.
One way of thinking about it is that instead of having the cost of information as an exogenous parameter, it is now endogenous which depends on the opponent's strategy.

The most closely related papers are 
\cite{Shahanaghi2022}, 
\cite{Ozdenoren2021} 
and
\cite{Bobtcheff2021}.
In the first two papers, the irreversible actions are observable. 
In \cite{Bobtcheff2021}, the safe action (`exit' in their paper) is observable with a positive probability.
I shut down the information externalities generated by observable actions to focus on the payoff externalities.
\cite{Shahanaghi2022}
discusses a dynamic preemption model with costless information where players have accuracy incentives and credibility concerns.
In equilibrium, players randomise between acquiring more information and taking an action due to the credibility concern and the observability of actions.
Preemption motives propagate the ex post errors in the actions.
This is different from my paper where
the randomisation in equilibrium is mainly due to the preemption motive.
The player in my model has no incentive to randomise if they were playing alone, which is not the case in 
\cite{Shahanaghi2022}.
\cite{Ozdenoren2021} 
studies a discrete-time experimentation model where players have incentives to preempt.
It shows that preemption motives caused by payoff externalities lead to less experimentation. 
With a discrete-time setting, they do not have an equilibrium in random stopping.
\cite{Bobtcheff2021}
investigates how publicity of actions affects the players' equilibrium strategy and payoff in a preemption game. 
In their model, similar to mine, the players have both the learning motive and the preemption motive.
The difference is that the risky action (`investment' in their model) is always observable while the safe action (`exit' in their model) can be private.
They argue that private signals and potential private actions create a winner's curse where the first risky action taker believes the opponent may have received a private breakdown and hence exited. 
To compensate, the player acquires more information to be more certain about the project before taking the risky action.
The main difference between this paper and mine is the unobservability of the irreversible action and the presence of breakthrough.
In my model, both the risky action, safe action and the signals are private.
Not being able to observe the risky action creates a  stronger preemption motive.
The players in my model tend to acquire less information because of this.

Another relevant strand of literature studies static information acquisition before a game 
(
see \cite{Hellwig2009}, \cite{Yang2015}, \cite{Han2018} and \cite{Denti2019}
).
Those papers discuss information acquisition before a coordination game.
When the players play a game with strategic complements, information choices exhibit strategic complementarity as well. 
In my model, information has the features of a strategic complement but could also be a strategic substitute.

This paper is also related to strategic experimentation literature that model similar R\&D races as a bandit problem (see \cite{Bolton1999StrategicExperimentation}, \cite{Keller2005StrategicBandits}, and \cite{Keller2010StrategicBandits}).
Besides the differences in the model setup introduced earlier,
there are three other differences.
First, strategic experimentation literature investigates the free-ride problem when information is a public good.
This is mainly due to the observability of the actions as well as sometimes the signals themselves.
In my model, information is private and information externality does not exist. 
There is no free ride.
Second, there is no exploration-exploitation tradeoff in my model because the players only take the action once and it is irreversible.
The players must stop acquiring information and take the risky action in order to exploit the outcome generated by the risky arm. 
Third, strategic experimentation literature studies the intensity of the experimentation.
In my model, I assume the players choose between to acquire information or not, but not the intensity at each time instant.
Since bang-bang solution is normally the optimal strategy in bandit problems, this simplification appears to be reasonable.

A distinct but related group of literature is about the equivalence between static and dynamic information acquisition.
\cite{Hebert2019} studies a dynamic rational inattention model and shows that the belief dynamics generated can resemble either diffusion processes or processes with large jumps.
\cite{Morris2019} studies what kind of static models with costly information acquisition has a sequential sampling foundation. 
In my model, because of the preemption feature,
dynamic information acquisition is intrinsic as a player's payoff depends on the order that they act.

%\textbf{Bandit and common learning literature.}

% outline of the paper

\section{The model}
\label{sec:model}

\subsection{Model setup}

There are two players $i \in \{1,2\}$.
At the beginning of the game, an unknown, fixed, payoff-relevant state $\omega \in \{H,L\}$ is drawn.
At any time $t\in[0,\infty)$, each player can take an irreversible action $x \in \{S(afe), R(isky)\}$, or delay and acquire information about the state.
The irreversible action gives the player an one-off payoff at the moment she takes the action.
Action $S$ yields a payoff which is normalised to be $0$.
Action $R$ payoff depends on the state and whether the player is the first or second to take $R$.
In state $\omega$,
the first player to take $R$ gets 
$
u_{\omega}\in \mathbb{R}
$
and the second gets 
$
u_{\omega}-\bar{\triangle}_{\omega}.
$
If the players take $R$ simultaneously, the payoff is
$
u_{\omega}-\underline{\triangle}_{\omega}.
$
At each time, if the player delays her action and acquires information, she incurs a positive information flow cost 
$
c
$
per unit of time.
I assume no time discounting.
Payoffs satisfy the following two assumptions.

\begin{assumption}
\label{ass:first>simul>second}
$
\bar{\triangle}_{\omega}>\underline{\triangle}_{\omega}>0
$
for 
$\omega \in \{H,L\}$.
\end{assumption}

\Cref{ass:first>simul>second} says that first $R$ taker gets a higher payoff than the second $R$ taker in both states.
If the players take $R$ simultaneously, the payoff is in between.

\begin{assumption}
\label{ass:RactionisbetterinstateH}
$
u_{L}
<
0
<
u_{H}-\bar{\triangle}_{H}.
$
\end{assumption}

\Cref{ass:RactionisbetterinstateH} says that $R$ yields a higher payoff than $S$ in state $H$ and a lower payoff in state $L$.
Furthermore, being the second to take $R$ in state $H$ is still better than taking $S$.
Players' incentive to be the first is increasing in the value of $\bar{\triangle}_{H}$.
The difference between the first and second $R$ taker payoff describes the intensity of the competition.
\Cref{ass:RactionisbetterinstateH} describes a gentle competition in the sense that being the second to take $R$ in state $H$ is not too bad.
This assumption is dropped in \Cref{subsec:competition}, 
where the competition is more intense as the second $R$ taker gets a lower payoff than taking $S$.

Before taking the irreversible action, players have access to costly information about the state.
Information is modelled using Poisson signals.
If a player acquires information for a short time period 
$
dt>0
$, 
then, in state 
$H$ ($L$, resp), 
she receives an 
$H$-state ($L$-state, resp)
revealing signal with rate $adt$ ($bdt$, resp).
Player $i$'s belief $p_{t}^{i}$ is the probability that the state is $H$.
I assume the players have a common prior $p_0$ and that they observe neither the opponent's action nor their signals.
At each time $t$, players update their beliefs using Bayes' rule.
When a player acquires information, her belief jumps to $0$ after receiving an $L$-state revealing signal and jumps to $1$ after receiving an $H$-state revealing signal.
In the absence of the revealing signal, player $i$'s belief evolves according to
\begin{align}
\label{equ:belief}
\frac{dp_{t}^{i}}{dt}=(b-a)p_{t}^{i}(1-p_{t}^{i}).
\end{align}
In the following part of the paper, I assume $b>a>0$.
In the absence of a revealing signal, the player's belief drifts up.
This is the `no news is good news' environment.

\subsection{Strategies and equilibrium}
\label{sub:strategiesandequilibrium}

If a player receives an $H$-state ($L$-state, resp) revealing signal, it is optimal to stop acquiring the signal and take $R$ ($S$, resp) regardless of the opponent's behaviour.
Therefore, it is sufficient to describe players' strategy conditional on no arrival of a revealing signal.
A pure strategy (defined below) specifies the time, $T^{i}$, at which the player stops and which action, $x^{i}$, they take in the absence of the revealing signals.

\begin{definition}
\label{def:purestrategy}
Player $i$'s \textit{pure strategy $s^{i}$} is defined as 
$
\left(
T^{i},x^{i}
\right)
\in 
\mathbb{R}_{+} \times \{R,S\}.
$
\footnote{
The notation $\mathbb{R}_{+}$ denotes the set of non-negative real numbers. }
\end{definition}

A mixed strategy (defined below) specifies the probability that the player stops before time $t$ \textit{conditional on no revealing signal}.

\begin{definition}
\label{def:mixedstrategy}
Player $i$'s \textit{mixed strategy $\gamma^{i}$} is defined as two non-decreasing measurable functions
$
\left(\rho^{i},\sigma^{i}\right)
$
where
$
\rho^{i}:\mathbb{R}_{+}\rightarrow\left[0,1\right]
$
and 
$
\sigma^{i}:\mathbb{R}_{+}\rightarrow\left[0,1\right]
$
satisfy
$\rho^{i}\left(t\right)+\sigma^{i}\left(t\right) \leq 1$
for $\forall t \in \mathbb{R}_{+}$.
The first element 
$
\rho^{i}
$ 
is the probability that player $i$ stops and takes $R$ before or at time $t$ \textit{conditional on no revealing signal}.
The second element
$
\sigma^{i}
$
is the probability that player $i$ stops and takes $S$ before or at time $t$ \textit{conditional on no revealing signal}.
\end{definition}

Before defining the equilibrium, I first write down the player's expected payoff.
Player $i$'s expected payoff from taking action $x$ at time $t$ in state $\omega$,
$
\mathbb{E}^{\gamma^{j}}\left[u_{\omega}^{x} \mid \omega,t \right],
$
depends on player $j$'s strategy $\gamma^{j}$.
The randomness of the payoff from taking $R$ comes from both player $j$'s strategy and the randomness of the signal.
For example, the opponent using a pure strategy 
$
(0,R)
$
induces degenerate conditional distributions at each time $t$ such that player $i$'s action-$R$ payoff in state $\omega$ is $u_{\omega}-\bar{\triangle}_{\omega}$ with probability one.
Let 
$$
U_{x}^{i,\gamma^{j}}
\left(t\right)
:=
\mathbb{E}_{\omega \mid t}
\mathbb{E}^{\gamma^{j}}\left[u_{\omega}^{x} \mid \omega,t \right]
$$
be player $i$'s expected payoff from taking action $x$ at time $t$,
where 
the expectation $\mathbb{E}_{\omega \mid t}$ is taken over the distribution of the state given player $i$'s time $t$ belief.
Then, player $i$'s payoff from taking action $x^{i}$ at time $T^{i}$
is
\begin{align}
\label{eq:time0valuefunction}
    \intop_{0}^{T^{i}}\pi_{t}\left(p_{t}^{i}a
    \mathbb{E}^{\gamma^{j}}\left[u_{H}^{R} \mid H,t \right]
    %+\left(1-p_{t}^{i}\right)bu^{S}
    -c\right)dt
    +
    \pi_{T^{i}}
    U_{x^{i}}^{i,\gamma^{j}}\left(t\right)
\end{align}
where 
$
\pi_{t}=p_{0}e^{-at}+\left(1-p_{0}\right)e^{-bt}
$
is the probability of no revealing signal up to time $t$.
Let 
$
U^{i,\gamma^{j}}
\left(t\right)
:=
\max_{x^{i}}
U_{x^{i}}^{i,\gamma^{j}}
\left(t\right)
$ 
be player $i$'s payoff associated with taking the optimal irreversible action $x^{i}$ at time $t$.
The following defines the perfect Bayesian equilibrium in pure strategies.
The perfect Bayesian equilibrium in mixed strategies can be defined in a similar manner.

\begin{definition}
\label{def:pbe}
A perfect Bayesian equilibrium in pure strategies is a strategy profile 
$
\left(s^{i},s^{-i}\right)
$
and beliefs 
$
\left(\left(p_{t}^{i}\right)_{t\in\left[0,T^{i}\right]},\left(p_{t}^{-i}\right)_{t\in\left[0,T^{-i}\right]}\right)
$
such that for 
$i\in \left\{ 1,2\right\} $,
\begin{enumerate}
\item 
$
    x^{i}
    \in
    \arg\max_{\tilde{x}^{i}}
    U_{\tilde{x}^{i}}^{i,s^{-i}}
    \left(T^{i}\right);
$
\item
$
T^{i}
\in
\arg\max_{\tilde{T}^{i}}\intop_{0}^{\tilde{T}^{i}}\pi_{t}\left(p_{t}^{i}a
\mathbb{E}^{s^{-i}}
\left[u_{H}^{R} \mid H,t \right]
%+\left(1-p_{t}^{i}\right)bu^{S}
-c\right)dt
+
\pi_{\tilde{T}^{i}}
U^{i,s^{-i}}
\left(\tilde{T}^{i}\right)
$
;
\item 
The belief $p_{t}^{i}$ evolves according to (\ref{equ:belief}).
\end{enumerate}
\end{definition}

I seek symmetric equilibria in both pure and mixed strategies. 
In the following discussion, the superscript $i$ representing the player is ignored.

The player's strategy and the random arrival of the signal jointly determine the distribution over the action taken and timing. 
Before proceeding, I define the following two conditional probabilities that characterise the distribution over action taken and timing.
Let
$
F_{\omega}\left(t\right)
$
(
$
G_{\omega}\left(t\right)
$,
resp
)
be the probability that the player takes 
$
R
$ 
(
$
S,
$
resp
)
before or at time $t$ in state $\omega$.
Then,
$
1-
F_{\omega}\left(t\right)
-
G_{\omega}\left(t\right)
\geq
0
$
is the probability that player $i$ continues acquiring information at time $t$ in state $\omega$.

\section{An illustrative example}
\label{sec:example}

In this section, I use a simple two-period model to illustrate the tradeoff between information acquisition and preemption.
I show how players' \textit{learning motives} and \textit{preemption motives} depend on the opponent's strategy and the prior.
In equilibrium, mixed strategies create endogenous randomnesses that either deter learning or prevent the opponent from preemption.

At time $t=0$, players can choose to acquire a signal at cost $c>0$ or to take one of the actions.
At time $t=1$, players have to take one of the actions. 
If a player acquires a signal, in state $H$ ($L$, resp), a revealing signal arrives with probability $a$ ($b$, resp), and the belief $p_1$ jumps to $1$ ($0$, resp).
If she does not receive the revealing signal, then, her belief is updated to
$
\frac{p_{1}}{1-p_{1}}=\frac{1-a}{1-b}\frac{p_{0}}{1-p_{0}}
$.
Acquiring a signal at time $0$ allows the player to learn the state and hence take the `correct' action ($R$ in state $H$ and $S$ in state $L$).
This gives the player the `learning motive'.
Not acquiring a signal, however, secures the player the first prize.
This gives the player the `preemption motive'.

The \textit{learning motive} is stronger when the prior is in the intermediate range and when the opponent takes $S$ at time $0$.
When the prior is in the intermediate range, the player is uncertain and hence has stronger incentives to learn.
When the opponent takes $S$ at time $0$, the player is the single decision maker in this game. 
The payoff associated with taking the correct action is the highest and hence the value associated with information is higher at each prior.
The \textit{preemption motive} is stronger when the opponent acquires a signal at time $0$.
This is because when the opponent acquires a signal, by taking  $R$ at time $0$, the player can secure herself the first $R$ taker payoff, while acquiring a signal at time $0$ decreases the probability of being the first  $R$ taker and hence decreases her expected payoff from taking  $R$.
The information becomes less valuable.
When the opponent takes  $R$ at time $0$, the player's incentive to preempt is less strong.
It is then optimal for the player to acquire a signal for a larger range of priors.
This is because she cannot preempt her opponent only to match their action.
This matching reduces the payoff from acting at time $0$ and correspondingly increases the payoff of waiting.

\begin{figure}
\includegraphics[width=\textwidth]{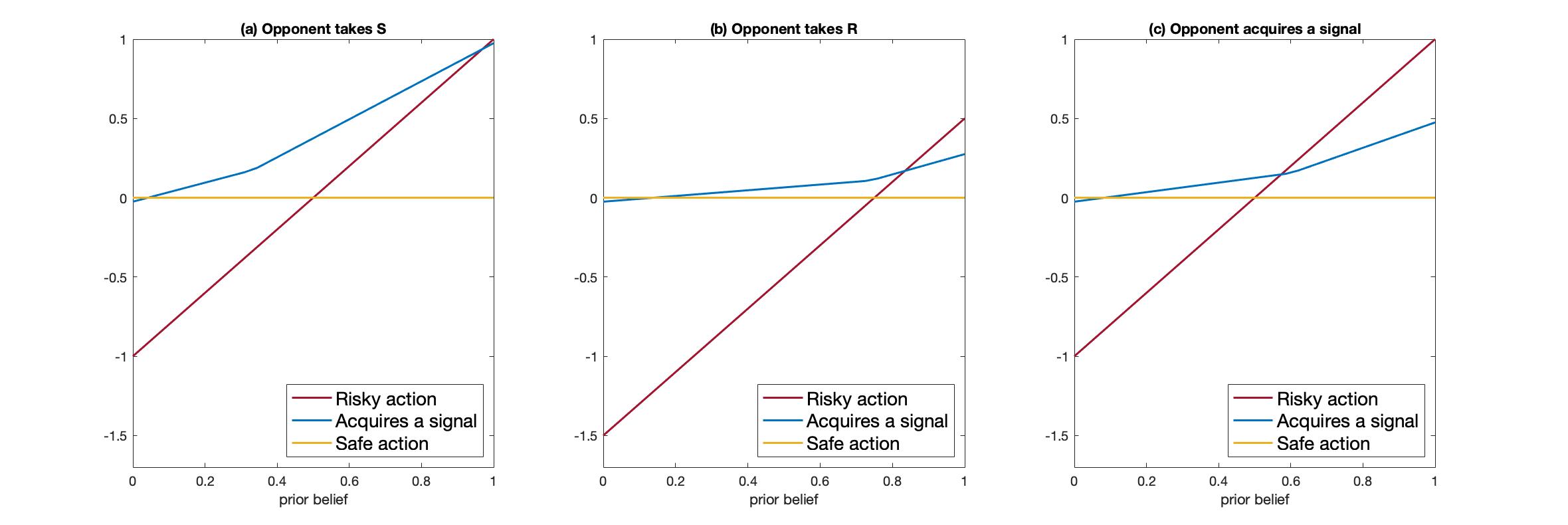}
\caption[A player's payoffs given opponent's different strategies]%
{
The player's payoffs given opponent's different strategies \par \small
Notes: This graph is drawn given the following parameter values: 
$u_{H}=1$,$u_{L}=-1$, $\bar{\triangle}_{\omega}=0.7$, $\underline{\triangle}_{\omega}=0.5$, $a=0.6$, $b=0.8$, $c=0.025$.
In all three panels, the red line represents the player's payoff from taking $R$ at time $0$, the blue line represents the player's payoff from acquiring a signal at time $0$, and the yellow line represents the player's payoff from taking $S$ at time $0$.
}
\label{fig:two_period_example}
\centering
\end{figure}

\Cref{fig:two_period_example} plots her payoffs associated with taking immediate actions and acquiring a signal:
Panel (a) is the player's payoffs when the opponent takes $S$ at time $0$, (b) is when the opponent takes $R$ at time $0$ and (c) is when the opponent acquires a signal at time $0$.
For the opponent's different strategies,
qualitatively, 
the player's best responses have similar properties: to take an action at time $0$ when the prior belief is extreme and to acquire a signal when the prior belief is in the intermediate range.
The difference is the range of the priors at which her best response is to acquire a signal.
This range of priors is largest when the opponent takes  $S$ at time $0$ and is smallest when the opponent acquires a signal.

Next, I briefly discuss the possible equilibria: the equilibrium where players preempt $R$, the equilibrium where players acquire a signal with positive probability, and the equilibrium where players randomise between two actions without acquiring a signal.

When the prior is sufficiently high, the preemption motive dominates and hence there is unravelling. 
The symmetric equilibrium is such that both players take immediate  $R$ at time $0$ without acquiring any signal.
The weak learning motive and the strong preemption motive work in the same direction which pushes the players to take $R$ immediately.
For intermediate priors, there exists an asymmetric equilibrium where one player acquires a signal and the opponent takes $R$ at time $0$.
In symmetric equilibrium, the players randomise over these two roles.
This is referred to as a \textit{random stopping strategy}: the player randomises between stopping (to take $R$) at time $0$ and time $1$.
By using this strategy, the player creates an endogenous uncertainty that prevents the opponent from preempting.
When a player acquires a signal with probability one, her opponent has incentives to preempt and it is easy for them to do so.
But when a player randomises, not only the value associated with preemption is reduced, it is also harder for her opponent to preempt because of the endogenous randomness.
For some intermediate priors, 
there exists another symmetric equilibrium where players randomise between two immediate actions without acquiring any signal.
This kind of randomisation deters learning.
For those intermediate priors, players are uncertain about the state.
By simply randomising between $R$ and  $S$ at time $0$, the player `hedges' against the uncertainty without paying extra information cost and at the same time, reduces her opponent's value associated with learning.
\footnote{The player taking $R$ at time $0$ with a positive probability reduces her opponent's expected payoff from taking  $R$ at time $1$.}
When this value is sufficiently low, the opponent's learning is deterred.

This two-period example shows the most important tradeoff in the model: the incentive to learn and the fear of being preempted.
However, it can only discuss \textit{whether or not} the players acquire information, but not \textit{how much} information they get.
To understand the optimal quantity of information the players acquire before taking an action, the dynamic model with multiple periods is of interest.
Next, I analyse the dynamic model introduced in \Cref{sec:model}.

%%%%%%%%%%%%%%%%%%%
%%%%%%%%%%%%%%%%%%%
%%%%%%%%%%%%%%%%%%%

\section{Equilibrium analysis}
\label{sec:equilibrium_analysis}

 \subsection{Single decision maker benchmark
}
\label{subsec:singleDM}

Before getting into the detail of the game, as a benchmark, I first consider the model with one single decision maker (DM).
This single DM model is a well-studied sequential sampling model due to \cite{Wald1945,Wald1947} where the DM chooses an optimal stopping time based on the samples she has observed.

The optimal stopping rule depends on the cost of information and the belief.
If the cost of information is too high, then, it is optimal to take an immediate action based on her prior.
In this case, learning about the unknown state does not give the DM sufficiently high benefit to compensate for the high cost.
When the cost is sufficiently small, it is optimal to acquire information for intermediate beliefs and to take an immediate action for extreme beliefs.
The DM's optimal policy at each time $t$ only depends on the current belief $p_{t}$ but not time $t$ itself. 
This is because all the past information is summarised by the belief at time $t$ and the information cost in the past is sunk.
The DM acquires information if the marginal cost $c$ is smaller than the marginal benefit.
The marginal benefit is higher when the DM is uncertain about the state.
Hence, it is optimal to take $S$ if the belief is sufficiently small, to acquire information if the belief is in an intermediate range, and to take $R$ if the belief is sufficiently big.
This is summarised in the following proposition.

\begin{proposition}
\label{prop:singledmresult}
There exist cutoffs
$
\underline{p}<\bar{p} 
$
and 
$\bar{c}$
such that 
when $c\leq \bar{c}$, given belief $p$, the DM's optimal policy is 
%\begin{itemize}
 %   \item 
to take $S$ if $p\leq\underline{p}$;
 %   \item 
to acquire information if $\underline{p}<p<\bar{p}$;
 %   \item 
and to take $R$ if $p \geq \bar{p}$.
%\end{itemize}
\end{proposition}

In the single DM case, the optimal strategy is deterministic: 
the cutoff where the DM stops acquiring information is a constant and is uniquely pinned down by the parameter values. 
When she stops at the cutoff, the optimal action is $R$ because she is convinced that the state is more likely to be $H$.
When the DM uses the optimal strategy, 
if she acquires information, she only makes mistakes in the low state: she correctly takes $R$ in high state with probability one and incorrectly takes $R$ in the low state with a positive probability.
This is because when acquiring information, she either receives a revealing signal and then stop or she acquires information until the belief drifts up to the upperbound $\bar{p}$.
In the high state, the DM takes the correct action $R$ in both events.
In the low state, she takes the incorrect risky action if she stops at the upperbound $\bar{p}$.

%The single DM's problem is a special case of the player's problem in the two-player game when the opponent takes immediate $S$.

\subsection{Properties of equilibrium}

Now I turn to the model with two players.
I establish the properties of the equilibrium strategy.
I explain why there are \textit{no jumps} in the player's equilibrium strategy at all $t>0$.
If there \textit{is} a jump in the player's strategy, it only happens at $t=0$.
This property is driven by preemption and the unobservability of actions.

\begin{lemma}
\label{lemma:rhocontinuous}
In equilibrium, 
$\rho\left(t\right)$,
the probability that the player stops and takes $R$ before or at time $t$ conditional on no revealing signal,
is continuous at all $t>0$.
\end{lemma}

In the absence of the revealing signal, the player does not take $R$ with positive probability mass at any time $t>0$.
This implies that the single DM's \textit{deterministic} optimal strategy is not part of the equilibrium in the two-player game.
The reason is twofold.
First, the jump at any $t>0$ in the opponent's strategy gives the player an opportunity to preempt and hence a profitable deviation.
If the opponent takes $R$ at some $\tau>0$ with a positive probability, 
then,
the player is betteroff stopping at $\tau-dt$
where $dt>0$ is infinitesimal.
At time $\tau-dt$, if the player continues acquiring information for $dt$ longer, 
she gains nothing
but loses the probability of being the first $R$ taker.
Second, the private action and the private breakdown give rise to \textit{winner's curse}.
It is more likely for the player to be the first $R$ taker in state $L$ than in state $H$ because the opponent might have received the breakdown and dropped out.
In other words, it is more likely to `win' in state $L$.
If the opponent takes $R$ at some $\tau>0$ with a positive probability, then,
conditional on no revealing signal and `winning', the value associated with taking $R$ at time $\tau$ has a downward jump.
It is then not a best response for the player to take $R$ at time $\tau$.

In equilibrium, $\rho$ can have a jump at $t=0$.
This means that in the absence of the revealing signal, the player takes $R$ with positive probability mass only at time zero.
This is a result of preemption.
The unravelling comes from the player undercutting the time of the opponent taking $R$ until time zero at which there is no more room for preemption.

%The second element of the player's strategy, $\sigma$, is not necessarily continuous in equilibrium.

\subsection{Equilibrium strategies}
\label{subsec:equilibrium_strategy}

This section introduces five strategies that appear in equilibrium.

\paragraph{The Immediate Action Strategy}
The \textit{Immediate $R$} (\textit{Immediate $S$}, resp) \textit{Strategy} is the pure strategy $\left(0,R\right)$ ($\left(0,S\right)$, resp) where the player does not acquire information and takes $R$ ($S$, resp) at time $0$.
%When the prior is such that the players are relatively certain about the state, or when the cost of signal is sufficiently high, there exists the equilibrium where both players use the \textit{Immediate action $R/S$ Strategy}.

%The \textit{Immediate Mix with No Learning Strategy} is a mixed strategy $\left(\rho^{IM},\sigma^{IM}\right)$ 
%where
%$\rho^{IM}\left(0\right)=1$
%and 
%$\sigma^{IM}\left(0\right)\in (0,1)$.
%This is a strategy such that the player does not acquire the signal and randomises between actions $R$ and action $S$ with probability $\sigma^{IM}\left(0\right)$ and $1-\sigma^{IM}\left(0\right)$ at time $0$.

\paragraph{The Immediate Mix Strategy}
This is a strategy where the player does not acquire information.
She randomly takes $R$ and $S$ at time zero.
The following is the formal definition.
The \textit{Immediate Mix Strategy} is a mixed strategy $\left(\rho^{IM},\sigma^{IM}\right)$ 
where
$
\rho^{IM}\left(t\right)
=
\rho^{IM}\left(0\right)\in (0,1)
$
for 
$
\forall t>0,
$
$
\sigma^{IM}\left(t\right)
=
\sigma^{IM}\left(0\right)\in (0,1)
$
for 
$
\forall t>0,
$
and 
$
\rho^{IM}\left(0\right)+\sigma^{IM}\left(0\right)=1.
$

%The \textit{Random Stopping Strategy} is a mixed strategy $\left(\rho^{RS},\sigma^{RS}\right)$ where 
%\begin{align*}
%\rho^{RS}\left(t\right)
%=
%\begin{cases}
%0 
%& t\leq\bar{T}^{RS}\\
%\Gamma^{RS}\left(t\right) 
%& \bar{T}^{RS}<t<\tilde{T}^{RS}\\
%1 
%& t\geq\tilde{T}^{RS}
%\end{cases}
%\end{align*}
%is continuous 
%and 
%$\Gamma^{RS}:(\bar{T}^{RS},\tilde{T}^{RS}) \rightarrow (0,1)$
%is an increasing function.
%The second element $\sigma^{RS}$ satisfies 
%$
%\sigma^{RS}\left(t\right)=1
%$
%for $t>\bar{T}^{RS}$.
%This is a strategy such that in the absence of the revealing signal, the time at which the player stops is randomised, and
%conditional on stopping, action $R$ is taken with probability one.

\paragraph{The Random Stopping Strategy}
This is a strategy where the player acquires information for a period of time and then stops with a positive rate at each time.
The action associated with stopping is $R$.
The following is the definition.
The \textit{Random Stopping Strategy} is a mixed strategy 
$
\left(\rho^{RS},\sigma^{RS}\right)
$
such that 
$
\rho^{RS}\left(\cdot\right),
$
the probability of taking $R$ in the absence of the revealing signal,
weakly increases and 
$
\sigma^{RS},
$
the probability of taking $S$ in the absence of the revealing signal,
equals zero for all $t\geq 0$.
The first element
$
\rho^{RS}\left(\cdot\right)
$
satisfies
$
\rho^{RS}\left(0\right)=0
$
and 
$
\rho^{RS}\left(t\right)=1
$
for 
$
t\geq\bar{T}^{RS}
$
where $\bar{T}^{RS}>0$.
When using this strategy, conditional on arriving at time $t$ 
\footnote{That is, the player receives no revealing signal and the player's strategy does not prescribe stopping.}
, the player stops and takes action $R$ with rate 
$
\frac{\frac{d\rho^{j}\left(t\right)}{dt}}{1-\rho^{j}\left(t\right)}.
$
She only stops and takes action $S$ after observing a breakdown.

\paragraph{The Mixed Learning Strategy}
This is a strategy where the player mixes between the \textit{Immediate $S$ Strategy} and the \textit{Random Stopping Strategy} at time $0$.
The following is the definition.
The \textit{Mixed Learning Strategy} is a mixed strategy 
$\left(\rho^{ML},\sigma^{ML}\right)$
such that 
$
\rho^{ML}\left(\cdot\right)
$
weakly increases and 
$
\sigma^{ML}\left(t\right)
=
\beta 
\in
\left(0,1\right)
$
for $\forall t$.
The probability that the player stops and takes $R$ before time $t$ 
(i.e.
$
\rho^{ML}\left(\cdot\right)
$
)
satisfies
$
\rho^{ML}\left(0\right)=0
$
and 
$
\rho^{ML}\left(t\right)=1-\beta
$
for 
$
t\geq\bar{T}^{ML}
$
where $\bar{T}^{ML}>0$.

\subsection{Main results}
\label{subsec:main_result}

The theorem below lists the symmetric equilibria given different priors.
Let 
$
\frac{p^{M}}{1-p^{M}}
:=
\frac{-\left(u_{L}-\underline{\triangle}_{L}\right)}{u_{H}-\underline{\triangle}_{H}}
$,
$
\frac{p^{L}}{1-p^{L}}
:=
\frac{-u_{L}}{u_{H}}
$
and 
$
\frac{\tilde{p}}{1-\tilde{p}}
:=\frac{-bu_{L}-c}{a\bar{\triangle}_{H}+c}
$
be three prior cutoffs.
Both $p^{M}$ and $p^{L}$ are positive as $u_{L}$ is negative.
A sufficiently small information cost $c$ guarantees that 
$
p^{M}<\tilde{p}
$

\begin{theorem}
\label{thm:mainresult}
When $b<2a$ and $c$ is sufficiently small, there exist cutoffs 
$
\underline{p}<p^{*}<p^{L}<p^{M}<\tilde{p}
$
such that:

If $p_{0}\leq \underline{p}$, there exists an equilibrium where both players use the Immediate $S$ Strategy;

If $\underline{p}<p_{0}<p^{*}$, there exists an equilibrium where both players use the Mixed Learning Strategy;

If $p^{*}\leq p_{0}<\tilde{p}$, there exists an equilibrium where both players use the Random Stopping Strategy;

If $p^{L}<p_{0}<p^{M}$, there exists an equilibrium where both players use the Immediate Mix Strategy;

If $p_{0}\geq p^{M}$, there exists an equilibrium where both players use the Immediate $R$ Strategy.
\end{theorem}

The lowest cutoff $\underline{p}$ is the lowerbound in \Cref{prop:singledmresult}, below which the single DM's optimal strategy is to take $S$ immediately.
The cutoff
$p^{*}$ is a fixed point and the detail can be found in \Cref{appen:proofofmainresult}.
The cutoff $p^{L}$ ($p^{M}$, resp) is the belief at which the player is indifferent between taking $S$ and taking immediate $R$ if she is the first $R$ taker (if she and the opponent take $R$ simultaneously, resp).

When the prior is extreme, both players use the \textit{Immediate $R$/$S$ Strategy} in equilibrium.
It seems to be intuitive as the players' learning motive is weak due to the relatively small value of information.
But this is not the whole story.
The subtlety is the role of the preemption motive. 
When the prior is sufficiently small, the preemption motive is absent because $S$ is the optimal action that gives the player the same payoff.
The weak learning motive is indeed the only reason why the players take immediate $S$.
However, when the prior is high, the preemption motive is the main reason why the players take immediate $R$ with a positive probability.
To understand the existence of the equilibria where the players use the immediate $R$ or the immediate mix strategy,
consider a static game where the players do not have access to information
\footnote{This is a game where both players have two actions: action $R$ and action $S$. 
In state $\omega$, if both players take action $R$ (action $S$, resp), the payoff is $u_{\omega}-\underline{\triangle}_{\omega}$ ($0$, resp) for each of them. 
If one player takes action $R$ and the other player takes action $S$, the player who takes action $R$ gets a payoff of $u_{\omega}$ and the player who takes action $S$ gets a payoff of $0$.}.
Consider the sufficiently high prior such that there exists an equilibrium where the players take $R$ with a positive probability in this static game
\footnote{That, when the belief is higher than $p^{L}$.}
.
Then, even if the players are now given access to costly information, it is still optimal to take an immediate action.
This is because when the prior is sufficiently high, the preemption motive is strong and the learning motive is weak.
In the event that the opponent takes immediate $R$, if the player also takes immediate $R$, she then gets the simultaneous-move payoff.
If the player acquires information for an infinitesimal time period $dt$, then, 
the loss from becoming a second $R$ taker is strictly positive but the gain from acquiring information is negligible.
The positive probability of the opponent taking immediate $R$ induces a downward jump in the player's expected payoff from taking $R$.
This deters learning.

\begin{lemma}
\label{lemma:comparisonoftwoupperbounds}
When $c$ is sufficiently small, 
$p^{M}<\bar{p}$.
\end{lemma}

The cutoff $\bar{p}$ is the upperbound in \Cref{prop:singledmresult}: the belief above which the player takes $R$ immediately.
\Cref{lemma:comparisonoftwoupperbounds} says that the range of priors where the players take immediate $R$ is larger when the players have the incentive to preempt.
The preemption motive leads to less information acquisition and more immediate risky action for relatively optimistic prior beliefs.

The intuition of the equilibrium where the players use the \textit{Random Stopping Strategy} can be explained by the tradeoff between learning and the preemption motive.
When the prior is in the lower intermediate range, the player has the learning motive because the value of information is high.
If a player were playing this game alone, the optimal strategy is to acquire information until the belief drifts up to an upperbound and then takes $R$ (see \Cref{prop:singledmresult}).
Before she stops acquiring information, the distribution of the time at which she takes $R$ is the distribution of the Poisson breakthrough, which is continuous.
At the time the belief drifts up to the upperbound, according to the strategy, the player stops and takes $R$.
Therefore, the induced distribution of the time at which she takes $R$ has an atom at the time when the belief drifts up to the upperbound.
If player $i$ were playing against player $j$ who uses this strategy, player $i$ would preempt player $j$.
To avoid being preempted by the opponent, what the player could do is to randomise at which time she stops and takes action $R$.
This randomisation will make the opponent's marginal cost and marginal benefit (explain later) from acquiring information the same and hence eliminate the incentive to preempt.
The indifference condition gives rise to an ODE that the equilibrium strategy satisfies which is concluded in \Cref{lemma:randomisedstoppingequilibriumcondition} in the appendix.
Here in the main text, I give expressions of the marginal cost and marginal benefit from acquiring information and show how they are determined by the opponent's strategy.
Suppose player $j$ uses the mixed strategy 
$
\left(
\rho, \sigma
\right)
$
defined in \Cref{def:mixedstrategy}.
Suppose $\rho$ and $\sigma$ are differentiable.
Then,
\begin{align*}
F_{H} \left(t\right) 
& =
\intop_{0}^{t}
\left[e^{-as}\left(1-\rho \left(s\right)\right)\left(a+\frac{\frac{d\rho \left(s\right)}{ds}}{1-\rho \left(s\right)}\right)\right]ds 
\\
& = 
1-e^{-at}\left(1-\rho \left(t\right)\right) 
\end{align*}
is the probability that player $j$ takes $R$ before time $t$ in state $H$ and
\begin{align*}
F_{L} \left(t\right) 
& =
\intop_{0}^{t}
\left[e^{-bs}\left(1-\rho \left(s\right)\right)\frac{\frac{d\rho \left(s\right)}{ds}}{1-\rho \left(s\right)}\right]ds 
\\
& =
\intop_{0}^{t}\left[e^{-bs}\frac{d\rho \left(s\right)}{ds}\right]ds. \end{align*}
is the probability that player $j$ takes $R$ before time $t$ in state $L$.
At time $t$, player $i$ is indifferent between taking $R$ now and $dt$ later if
\begin{align}
\label{eq:marginalcostequalsmarginalbenefit}
c
+
p_{t}\frac{dF_{H}\left(t\right)}{dt}\bar{\triangle}_{H}
+
\left(1-p_{t}\right)\frac{dF_{L}\left(t\right)}{dt}\bar{\triangle}_{L}
=
\left(1-p_{t}\right)b\left[-u_{L}+F_{L}\left(t\right)\bar{\triangle}_{L}\right].
\end{align}
The left-hand side of (\ref{eq:marginalcostequalsmarginalbenefit}) is the marginal cost of acquiring information for $dt$ longer at time $t$ and the right-hand side is the marginal benefit.
The cost consists of the \textit{direct} information cost $c$ and the \textit{indirect} cost from being preempted. 
Intuitively, in state $\omega$, the opponent preempts the player in this short time period $dt$ with probability 
$
\frac{dF_{\omega}\left(t\right)}{dt}.
$
The benefits of acquiring information comes from the breakdown because it corrects the player's action from $R$ to $S$ in state $L$.
The player gets the payoff of zero instead of the negative payoff associated with $R$ in state $L$.
This indifference condition gives an ODE that the equilibrium strategy satisfies.
\Cref{lemma:randomisedstoppingequilibriumcondition} and \Cref{lemma:RSTequilibriumexistence} in \Cref{appen:proofofmainresult} characterises the conditions that the equilibrium strategy satisfies.

\begin{corollary}
\label{lemma:multipleequilibria} 
When $c$ is sufficiently small, there exist multiple equilibria if $\underline{p}<p_{0}<p^{M}$.
\end{corollary}

When $\underline{p}<p_{0}<\tilde{p}$, 
in the equilibrium where the players use the Random Stopping Strategy, the players both acquire information up to some time $\hat{T}$ and then randomly stop at each time greater than $\hat{T}$.
The time $\hat{T}$ at which the player starts the randomisation is not unique and this leads to multiple equilibria.
To get the intuition, suppose player $j$ uses the Random Stopping Strategy such that she starts the random stopping at some time $\hat{T}>0$.
To have both players using this strategy as an equilibrium, before time $\hat{T}$, player $i$ must strictly prefer to acquire information given that the opponent acquires information.
That is, given the opponent is still acquiring information, the marginal cost from acquiring information for $dt$ longer must be smaller than the marginal benefit.
As a result, the latest time instant at which the player starts randomising is the first time instant at which the marginal cost of acquiring information equals the marginal benefit given that the opponent acquires information.
Before this time point, if the opponent acquires the information, the marginal benefit of acquiring information exceeds the cost.
The player acquires information.
After this time point, even though the player knows the opponent is still acquiring information, the value associated with information is too small to keep the player engaged.
After time $\hat{T}$, player $i$ is indifferent between taking $R$ at each time $t>\hat{T}$, and she must prefer to randomise instead of taking $S$.
Therefore, the earliest time instant at which the player is willing to start the random stopping is the earliest time instant such that taking $R$ gives her a payoff of at least zero.
Before this time instant, the player is not optimistic enough about the state to take $R$.
After this time instant, the player matches the opponent's action.
If the opponent's strategy is to randomise, the player also randomises.

The intuition behind the multiplicity is the strategic complementarity of information.
Information is a strategic complement at time $t<\hat{T}$ because when the opponent is still acquiring information, the value of information is relatively high.
The player hence is also willing to acquire more information.
If the opponent starts randomising, then, the value associated with acquiring information becomes lower.
The player then matches the time at which the opponent starts randomising.
This strategic complementarity also explains the existence of multiple equilibria when 
$
p_{0} 
\in
\left[
p^{L},p^{M}
\right)
$.
For this range of priors, according to \Cref{thm:mainresult}, 
there exist both kinds of equilibria where the players use the Random Stopping Strategy and where the players use the Immediate Mix Strategy.
In other words, there simultaneously exist the equilibria where the players acquire information and where the players do not acquire information.

The existence of the equilibrium where both players use the \textit{Random Stopping Strategy} requires the player to start acquiring information at time $0$.
If the prior is relatively low, then, the expected duration of acquiring information before stopping is long and the cost of delay is high. 
In addition, the players are pessimistic about the state being high.
Therefore,
the players do not have strong incentives to start acquiring information and hence this symmetric equilibrium where both players acquire information do not exist.
However,
there exists asymmetric equilibrium where one player acquires information and the other takes the immediate safe action.
This is because when the player faces a longer potential delay of action, the expected benefit associated with taking $R$ must be higher to compensate for the expected information cost.
The benefit associated with taking $R$ depends on the opponent's strategy.
If the opponent's takes $S$ at time $0$, then, the player's gain from information is higher.
This higher gain can incentivise the player to start learning at time $0$.
However, if the opponent acquires information, then, the player gains less from taking $R$ because there is a positive probability that she can only get the lower payoff.
This gain may not be sufficient to incentivise the player to start acquiring information at time $0$.
In this case, information is a strategic substitute. 
In the symmetric equilibrium, the players mix between two strategies at time $0$: the Immediate $R$ Strategy and the Random Stopping Strategy.
Compared to the case that the opponent uses Random Stopping Strategy, the opponent takes $R$ with a lower probability.
This in turn increases the player's gain from taking $R$ and hence increases the gain from acquiring information.
Then, the player is incentivised to start learning at a lower prior.

\section{Information cost and the competition}

In this section, I discuss the effects of a vanishing information cost $c$ and the effects of the intensity of competition.

\subsection{Vanishing information cost $c$}

When the information cost vanishes, the players acquire information for a larger range of priors but not for all the priors.
In the equilibrium where the players use the random stopping strategy, they acquire information for a longer duration before randomisation.
When they randomise, they stop and take $R$ at a higher rate.

The cost of acquiring information consists of two parts: the exogenous information cost $c$ and the endogenous cost from being preempted.
This is the main difference between the two-player case and the single DM case.
In the single DM case, since the exogenous information cost $c$ is the total cost of acquiring information,
when $c$ vanishes,
the DM always acquires information until she is certain.
In contrast, 
in the two-player game, the total information cost does not vanish with the exogenous cost $c$.
When the prior is high, the player believes that the state is very likely to be $H$ and infers that the probability of being preempted is high.
Thus, the endogenous cost is high for optimistic priors even when $c$ vanishes.
As a result, when $c$ vanishes, the equilibrium such that both players acquire information exists for a larger range of priors, but does not exist for extremely high priors.
However, when the prior is low, the player is pessimistic about the state and hence she believes that it is not likely that the opponent takes $R$.
In addition, for low priors, if the player does not acquire information, the optimal action is to take $S$, the payoff of which does not depend on the opponent's action.
Thus,
the endogenous information cost is low when the prior is close to zero.
As a result, for low priors, the total information cost vanishes with $c$ and the player is willing to acquire information for low priors.
This property is summarised in the following proposition.

\begin{proposition}
\label{prop:howcutoffschangeswithcost}
The cutoff $\tilde{p}$ decreases in $c$. 
When 
$c \rightarrow 0$,
$\underline{p}\rightarrow 0$ 
and
$
\tilde{p}
\rightarrow
\frac{\frac{b}{a}\frac{-u_{L}}{\bar{\triangle}_{H}}}{1+\frac{b}{a}\frac{-u_{L}}{\bar{\triangle}_{H}}}
\in
\left(0,1\right)
$. 
\end{proposition}

For the equilibrium where the players use the random stopping strategy, when $c$ vanishes, the effects are twofold. 
First, the players acquire information for a longer duration before they start the randomisation stage.
Second, when the players just enter the randomisation stage, they stop and take $R$ at a higher rate.
The first is intuitive as when $c$ decreases, information becomes cheaper and hence players are willing to acquire information for a longer period.
The intuition for the second effect is related to the indifference condition when the players randomise.
In equilibrium, the players randomise between acquiring information and taking $R$ because the marginal cost and marginal benefit from acquiring information are the same. 
The marginal cost from acquiring information consists of the exogenous and the endogenous information cost.
When the player just enters the randomisation stage, as shown in (\ref{eq:marginalcostequalsmarginalbenefit}),
the marginal benefit is not affected by the exogenous information cost $c$ and the players' stopping rate.
The marginal benefit depends only on the cumulative density of the opponent taking $R$.
According the equilibrium strategy, since the opponent has not started the random stopping yet, the marginal benefit only depends on the arrival of the revealing signals but not player's stopping rate.
The endogenous information cost however, increases in the opponent's stopping rate.
If the players stop and take $R$ at a higher rate, the endogenous information cost is higher.
As a result, when the players just enters the randomisation stage, if the exogenous information cost $c$ vanishes,
the marginal benefit is unchanged while the marginal cost decreases because of the vanishing information cost $c$.
To have the indifference condition hold,
the players stop and take $R$ at a higher rate to increase the endogenous information cost.

\subsection{Competition}
\label{subsec:competition}

The competition in this game comes from the payoff difference between the first and second $R$ taker.
The intensity of the competition increases in the value of the payoff difference.
In all the previous discussion, the payoff difference satisfies \Cref{ass:RactionisbetterinstateH}.
That is, the second $R$ taker gets a higher payoff than taking $S$ in state $H$.
With this assumption, after the player learns the state is high, taking $R$ is a dominant strategy.
To understand the effect of an intense competition, I drop \Cref{ass:RactionisbetterinstateH} and impose \Cref{ass:secondRtakergetsnegativepayoff}.

\begin{assumption}
\label{ass:secondRtakergetsnegativepayoff}
$
u_{L}
<
u_{H}-\bar{\triangle}_{H}
<
0$.
\end{assumption}

\Cref{ass:secondRtakergetsnegativepayoff} says that in state $H$, the payoff of the second $R$ taker is smaller than the safe action payoff.
This assumption explicitly imposes a `loser gets punished' condition
\footnote{Note that given \Cref{ass:RactionisbetterinstateH}, the player may eventually be punished due to the total information cost. The difference here is that the player explicitly knows that the second prize from taking $R$ is worse than taking $S$.}.
Given \Cref{ass:secondRtakergetsnegativepayoff}, whether the player takes $R$ after learning the state is $H$ depends on her belief on the opponent's strategy.
Take an extreme case as an example.
If the player believes that the opponent takes $R$ immediately,
then, taking $S$ is the optimal action regardless of the state.
Information becomes worthless because learning the state does not change the player's action and $S$ always gives the player a constant payoff.
Then it is never optimal for the player to acquire information.
This deters information acquisition completely.

In a less extreme case, information acquisition is not necessarily deterred at the beginning.
Instead, there exists an upperbound on the duration of information acquisition.
Suppose the opponent uses some well-behaved strategy such that she acquires information at time $0$ with probability one.
Then, 
the player's expected payoff from taking $R$ conditional on state $H$ is the highest at time 0 and then decreases.
This is because the player gets the first $R$ taker payoff for certain at time $0$ and as time passes, the probability of the opponent taking $R$ increases.
Since the second $R$ prize is negative, there exists a time instant at which the player's expected payoff from taking $R$ conditional on state $H$ decreases to zero.
Then, despite learning the state or not, at and after this time instant, the player's optimal action is $S$.
The player hence gains nothing from acquiring information after this time instant.
However, this fact that information becomes worthless after certain time instant does not deter information acquisition at the beginning.
When the player just starts acquiring information,
the probability of being preempted is low. 
She puts a higher weight on the first prize conditional on state $H$.
The negative second prize does not matter too much.
Information acquisition is not deterred as long as the expected payoff from taking $R$ conditional on state $H$ is positive.
Once at the time instant such that $R$ gives the player an expected payoff of zero in state $H$, information acquisition stops and $S$ will be taken.
This puts an upperbound on the duration of information acquisition. 
Since the player takes the safe action that creates no preemption motive when stopping, there is no random stopping in equilibrium.
As a result,
there exists a pure strategy equilibrium where in the absence of the revealing signal, the players acquire information up to some time and then take the safe action.
This result is summarised in the following proposition.

\begin{proposition}
\label{prop:competition}
Suppose the payoff satisfies \Cref{ass:first>simul>second} 
and 
\Cref{ass:secondRtakergetsnegativepayoff}.
There exists a prior cutoff
$
p^{NR}
$
and a non-negative upperbound 
$
T^{PS}
=
\frac{1}{a}
\log
\frac{\bar{\triangle}_{H}}{-\left(u_{H}-\bar{\triangle}_{H}\right)}
$
such that
for the prior
$
p_{0}
\in
\left(
p^{NR},\tilde{p}
\right),
$
there exists a symmetric equilibrium where the players use the pure strategy 
$
\left(
T^{PS}, S
\right).
$
\end{proposition}

This proposition says that when the second prize for $R$ is negative, for intermediate priors, there exists an equilibrium where the players acquires information for a maximum of $T^{PS}$ time.
Before time $T^{PS}$, the players stop and take an action only after receiving the revealing signal. 
The maximum time $T^{PS}$ is independent of the prior.

The existence of such pure strategy equilibrium is the main difference between the cases with the positive and negative second prize in state $H$.
The intuition is that when the second $R$ prize in state $H$ is negative, the longer the player acquires information, the less attractive $R$ becomes.
When the player becomes sufficiently certain that the state is $H$, she is also convinced that $R$ has been taken.
At the end of the information acquisition stage, being sufficiently optimistic is associated with taking $S$.
However, in the positive second prize case, being sufficiently optimistic is associated with taking $R$, which generates the preemption motive and hence the equilibrium in the random stopping strategy.
In equilibrium, the maximum duration $T^{PS}$ is independent of the prior.
Given that the opponent takes $R$ after receiving the $H$-state revealing signal, $T^{PS}$ is the time at which the expected payoff from taking $R$ \textit{conditional} on state $H$ decreases to zero.
Since $T^{PS}$ is pinned down by the payoff in state $H$, it is independent of the prior.

Next I discuss the effect of a more intense competition while \Cref{ass:RactionisbetterinstateH} holds.
In this case, when the difference between the first and second prize from taking $R$ increases, the equilibrium where both players acquire information exists for a smaller range of priors. 
This is intuitive because the endogenous cost associated with acquiring information increases as the game becomes more competitive.
Then, the equilibrium where the players acquire information is harder to be sustained for the lower priors.

The discussion above is about the case when competition becomes more intense.
The other limiting case is when there is no competition.
That is, when the payoff difference between the first and second $R$ taker is zero.
The payoff from taking $R$ does not depend on the opponent's action.
Then, there is no strategic interaction between the two players.
The equilibrium in this limiting case is such that both players use the single DM optimal strategy.

\section{Extensions}

\subsection{More than two players}

In this section, I generalise the original model to $N>2$ players.
I assume that the first player to take $R$ gets the first prize and all other $R$ takers get the second prize.
I use this generalisation to show that as the number of players increases, the range of priors where the learning equilibrium can exist shrinks.

Increasing the number of players intensifies the competition in the game.
In \Cref{subsec:competition}, I discussed the competition in terms of the payoff difference between the first and second prize associated with $R$.
This section further discusses the effect of competition in terms of number of players in the game.
I show that the existence of the learning equilibrium requires the competition to be not too intense.
The payoff difference between the first and the second prize must decrease as fast as $\frac{1}{N-1}$ to guarantee the existence of the learning equilibrium.

Since I assume that the first $R$ taker gets the first prize and all other $R$ takers get the second prize, the player only cares about whether the first $R$ prize has been taken or not.
Because of this, from the player's point of view, the remaining $N-1$ players essentially act as one big opponent.
Let
$
\frac{\tilde{p}_{N}}{1-\tilde{p}_{N}}
:=
\frac
{-bu_{L}-c}
{\left(N-1\right)\bar{\triangle}_{H}a+c}
$
be a prior cutoff.
Since $u_{L}<0$, when $c$ is sufficiently small, $\tilde{p}_{N}$ is positive.
The following proposition characterises a necessary condition for the existence of the learning equilibrium where the players use the random stopping strategy.

\begin{proposition}
\label{prop:N-playerequilibrium}
For $N>2$, there exists an equilibrium where the players use the random stopping strategy only if 
$
p_{0}
<
\tilde{p}_{N}.
$
The cutoff $\tilde{p}_{N}$ decreases in $N$.
When $N\rightarrow \infty$, $\tilde{p}_{N} \rightarrow 0$.
\end{proposition}

The learning equilibrium can exist only if the prior is sufficiently low and this prior cutoff decreases in the number of players.
When the number of players goes to infinity, the cutoff $\tilde{p}_{N}$ approaches zero and the potential learning region vanishes.

\begin{corollary}
When $N\rightarrow \infty$, 
if $\bar{\triangle}_{N}$ decreases as fast as $\frac{1}{N-1}$, then, $\tilde{p}_{N}$ is finite.
\end{corollary}

When the number of players increases to infinity, if the prize difference $\bar{\triangle}_{H}$ decreases as fast as $\frac{1}{N-1}$,
then, the potential learning region still exists.
The intuition is that the learning equilibrium can exist only if the competition is not too intense.
If there are a lot of players competing, then, the payoff difference cannot be too big.

\subsection{Observable actions}
\label{sub:observableactions}

The irreversible risky and safe actions are assumed to be private in this paper.
This allows me to focus on the role of payoff externalities.
In this extension, I assume that after a player takes an irreversible action, it is immediately observed by the opponent.
The public actions generate information externalities as observing no action taken is itself informative.
The purpose of this extension is to show that the existence of the learning equilibrium found in \Cref{thm:mainresult} is robust to some exposure to information externalities.
To be more specific, I define \textit{the mimicking and random stopping strategy (MRSS)} and show that there exists a symmetric equilibrium where the players use this strategy.

\begin{definition}
The mimicking and random stopping strategy (MRSS) is a strategy such that
\begin{enumerate}
    \item After receiving the $H$-state revealing signal, the player stops and takes $R$ immediately;
    \item After receiving the $L$-state revealing signal, the player stops and takes $S$ immediately;
    \item After receiving no revealing signal and observing no action taken, the player stops and takes $R$ at each time $t$ with a rate $h\left(t\right)>0;$
    \item After receiving no revealing signal and observing the opponent taking $S$, the player stops and takes $S$;
    \item After receiving no revealing signal and observing the opponent taking $R$, the player uses the single DM optimal strategy.
\end{enumerate}
\end{definition}

The MRSS and the random stopping strategy are similar in the sense that after observing no private revealing signal, the player randomly stops and takes $R$. 
The conditions for the existence of the symmetric equilibrium where the players use MRSS is presented in the following proposition.

\begin{proposition}
\label{prop:observableactionresult}
Suppose
$
\bar{\triangle}_{H}
=
\bar{\triangle}_{L}
$
and $c$ sufficiently small such that 
$
p^{L}<\tilde{p}.
$
If
$p^{L}<p_{0}<\tilde{p}$,
then,
there exists an equilibrium where the players use MRSS.
\end{proposition}

\Cref{prop:observableactionresult} suggests that when there are information externalities,
there still exists the equilibrium where the players randomly stop and take $R$ after no revealing signal.
The reason is twofold. 
First, the observability of the actions does not eliminate the preemption motive.
After the history of no action taken, if the player's strategy prescribes taking $R$ with positive probability mass, 
then, the opponent has an incentive to preempt.
Second, random stopping reduces the information involved in the player's action.
The action taken (or no action taken) contains the player's private information.
It is essentially an additional signal and the informativeness of it depends on the player's strategy.
For example, if the player only stops acquiring information after receiving a revealing signal,
then, her action perfectly reveals her private information.
In this case, this additional signal is very informative for the opponent.
However, if the player stops acquiring information and takes an action randomly, then, her action contains less private information.
Information externalities enhance the player's incentive to stop randomly.

\section{Conclusion}

In this paper,
I study a model in which the players can acquire costly private information before taking an irreversible private action.
Acquiring information takes time which allows the player to take the `correct' action but increases the probability of being preempted.
With the assumption that there is a first-mover advantage associated with the risky action, 
I find that in the equilibrium where the players acquire information, they become more optimistic that the state is high but at the same time, the conditional expected payoff from taking the risky action decreases.
Because of the interaction between the learning motive and the preemption motive, the players randomly stop and take the risky action.
This result is significantly different from the single decision maker case where the optimal strategy is deterministic.

In my model, depending on the prior, the players' decisions on acquiring information can be both strategic substitutes or strategic complements.
The strategic substitutability prevails when players have pessimistic priors and it induces players' initial mix between acquiring information and immediate exit.
The strategic complementarity prevails when players are relatively optimistic and it gives rise to the multiple equilibria where players use the random stopping strategy.
In addition, I find that the equilibrium where the players do not acquire information can exist for not only extreme priors but also a large range of intermediate priors.
This stems from the presence of the endogenous information cost.
Such equilibrium exists for arbitrarily small exogenous information cost.

\bibliography{references}
\bibliographystyle{agsm}

\appendix

\section{Formulation of the player's problem}
\label{sec:playersproblem}

This section formulates one player's problem given opponent's strategies.
Given player $j$'s strategy $\gamma^{j}$, player $i$'s best-reply problem is to choose a time $T^{i}$ with value function 
$
\hat{V}^{i,\gamma^{j}}:\mathbb{R}_{+}\rightarrow \mathbb{R}
$ 
such that
\begin{align}
\hat{V}^{i,\gamma^{j}}\left(0\right)
&:=
\max_{T^{i}}\intop_{0}^{T^{i}}\pi_{t}\left[p_{t}^{i}a\mathbb{E}^{\gamma^{j}}\left[u_{H}^{R}\mid H,t\right]-c\right]dt+\pi_{T^{i}}U^{i,\gamma^{j}}\left(T^{i}\right)
\\
&\text{s.t. }\frac{dp_{t}^{i}}{dt}=(b-a)p_{t}^{i}(1-p_{t}^{i}),
\nonumber
\end{align}
where 
$
\mathbb{E}^{\gamma^{j}}\left[u_{H}^{R}\mid H,t\right]
$
and 
$
U^{i,\gamma^{j}}\left(T^{i}\right)
$
are as defined in \Cref{sub:strategiesandequilibrium}.
The Hamilton-Jacobi-Bellman (HJB) equation for player $i$'s problem is the following differential equation in 
$
V^{i,\gamma^{j}}:\mathbb{R}_{+}\rightarrow \mathbb{R},
$
where
\begin{align}
\label{eq:hjbequation}
\max
\left\{
\overbrace{
p_{t}^{i}a\left[\mathbb{E}^{\gamma^{j}}\left[u_{H}^{R}\mid H,t\right]
-
V^{i,\gamma^{j}}\left(t\right)\right]
+
\left(1-p_{t}^{i}\right)b\left[-V^{i,\gamma^{j}}\left(t\right)\right]
}
^{\text{A}}
+
\overbrace{
\frac{dV^{i,\gamma^{j}}\left(t\right)}{dt}
}
^
{\text{B}}
-c,\right.
\nonumber
\\
\left.
U^{i,\gamma^{j}}\left(t\right)
-
V^{i,\gamma^{j}}\left(t\right)
\right\}  
& =0
.
\end{align}
The interpretation of (\ref{eq:hjbequation}) is that at time $t$, player $i$ chooses between to continue acquiring information or stopping. 
She acquires information if the marginal gain is greater than the marginal cost.
Otherwise, she stops and gets the payoff 
$
U^{i,\gamma^{j}}\left(t\right)
$.
The marginal gain consists of the expected gain from receiving the revealing signal (labelled as A in (\ref{eq:hjbequation})) plus the rate of change of the value (labelled as B in (\ref{eq:hjbequation})).

Given opponent's strategy 
$
\gamma^{j}
$, 
if the player's problem is well-behaved, then, the value function 
$
V^{i,\gamma^{j}}\left(t\right)
$
is a classical solution to the HJB equation (\ref{eq:hjbequation}).
The player's best response can then be characterised correspondingly.
However, in our problem, (\ref{eq:hjbequation}) is not well-behaved because given the opponent's strategy 
$
\gamma^{j}
$,
$
U^{i,\gamma^{j}}\left(t\right)
$
has a kink and may not be continuous.

%This section shows that player $i$'s best response to player $j$'s mixed strategy consists of random stopping.
%This suggests the use of mixed strategies in equilibrium.
%Next section introduces the strategies that appear in equilibrium.

%When prior is in the intermediary range, player $i$'s best response is to acquire information at time $0$.
%When the belief drifts up, player $i$ acquires information at each time $t$ if the marginal gain is greater than the marginal cost.
%Given player $j$'s mixed strategy, the marginal gain and the marginal cost can be the same for a period of time.
%When this happens, player $i$ randomises between immediate stopping and acquiring information for $dt$ longer.
\section{Proof of \Cref{prop:singledmresult}}
\label{appen:proofofsingleDMresult}

When player $i$ is a single DM, her problem is 
\begin{align*}
\hat{V} \left(p_{0}\right)
:=
\max_{T }\intop_{0}^{T }\pi_{t}\left[p_{t} a u_{H} +\left(1-p_{t} \right)bu^{S}-c\right]dt+\pi_{T }U \left(T \right)
.
\end{align*}
The HJB equation is
\begin{align*}
\max\left\{ p_{t} a\left[u_{H} - V \left(t\right)\right]+\left(1-p_{t} \right)b\left[u^{S}-V \left(t\right)\right]+\frac{dV \left(t\right)}{dt}-c,\right.
\nonumber
\\
\left.U \left(t\right)-V \left(t\right)\right\} =0
\end{align*}
Since the argument $t$ only enters the equation via $p_{t}$, I use $p$ instead of $t$ as the state variable.
Then HJB equation becomes
\begin{align}
\label{eq:singleDMHJB}
\max\left\{ p a\left[u_{H}-V \left(p \right)\right]+\left(1-p \right)b\left[u^{S}-V \left(p \right)\right]+\frac{dV \left(p \right)}{dp }\frac{dp }{dt}-c,\right.
\nonumber
\\
\left.U \left(p \right)-V \left(p \right)\right\} =0
\end{align}
To find player $i$'s value function, I construct a candidate value function and show it is a viscosity solution of (\ref{eq:singleDMHJB}).

If the learning region (the range of beliefs at which the DM acquires the signal) exists, the value function is a solution of the ordinary differential equation
\begin{align}
\label{eq:odehjb}
pa\left[u_{H}-V\left(p\right)\right]+\left(1-p\right)b\left[u^{S}-V\left(p\right)\right]+\frac{dV\left(p\right)}{dp}\left(b-a\right)p\left(1-p\right)
=
c
.
\end{align}
The free boundary solution to (\ref{eq:odehjb}) is 
\begin{align*}
V^{L}\left(p\right)
=
p\left[u_{H}-\frac{c}{a}\right]+\left(1-p\right)\left[u_{L}-\frac{c}{b}\right]+K\left(\frac{p}{1-p}\right)^{\frac{b}{b-a}}\left(1-p\right)
\end{align*}
where $K$ is a constant.
Suppose the learning region is $\left(\underline{p},\bar{p}\right)$.
Value matching and smooth pasting pin down the value of $\bar{p}$ and $K$.
Then, value matching pins down $\underline{p}$.
We have
\begin{align*}
\frac{\bar{p}}{1-\bar{p}}
=
\frac{u^{S}-u_{L}-\frac{c}{b}}{\frac{c}{b}}
:=
\bar{L},
\end{align*}
\begin{align*}
K
=
\frac{c}{b}\left(\frac{b}{a}-1\right)\bar{L}^{1-\frac{b}{b-a}}
\end{align*}
and
$\underline{L}:=\frac{\underline{p}}{1-\underline{p}}$ satisfies
\begin{align}
\label{equ:p_lowerbardefinition}
\left[u_{H}-u^{S}-\frac{c}{a}\right]\underline{L}+K\underline{L}^{\frac{b}{b-a}}
=
\frac{c}{b}.
\end{align}
Let 
$\hat{p}$ be a belief cutoff at which the player is indifferent between $R$ and $S$.
The existence of the learning requires $\bar{p}>\hat{p}$.
That is,
$$
c
<
b\frac{\left(-u_{L}\right)u_{H}}{u_{H}-u_{L}}
:=
\bar{c}
.
$$
Outside the learning region, the value function satisfies 
$
V\left(p\right)=U\left(p\right)
$.
When $c<\bar{c}$, the candidate value function is 
\begin{align*}
V\left(p\right)
=
\begin{cases}
p\left[u_{H}-\frac{c}{a}\right]+\left(1-p\right)\left[u_{L}-\frac{c}{b}\right]+K\left(\frac{p}{1-p}\right)^{\frac{b}{b-a}}\left(1-p\right) 
& p\in\left(\underline{p},\bar{p}\right)
\\
U\left(p\right) 
& \text{o.w.}
\end{cases}    
.
\end{align*}
This candidate has a kink at $\underline{p}$ and is differentiable everywhere else.
I next show that it is a viscosity solution of (\ref{eq:singleDMHJB}).
Let 
\begin{align*}
H\left(p,V\left(p\right),V'\left(p\right)\right)
:=
pa\left[u_{H}-V\left(p\right)\right]+\left(1-p\right)b\left[u^{S}-V\left(p\right)\right]+V'\left(p\right)\left(b-a\right)p\left(1-p\right)-c.
\end{align*}
For the points where $V\left(p\right)$ is differentiable,
I show that 
(1) if $p>\bar{p}$, then, $H\left(p,V\left(p\right),V'\left(p\right)\right)\leq0$;
(2) if $p \in\left(\underline{p},\bar{p}\right]$, then,
$V\left(p\right)\geq U\left(p\right)$;
(3) if $p<\underline{p}$, then $H\left(p,V\left(p\right),V'\left(p\right)\right)\leq0$.
At the point where $V\left(p\right)$ is not differentiable, that is, at $p=\underline{p}$,
I show that $H\left(p,V\left(p\right),z\right)\geq0$ for $z\in D^{+}$ where $D^{+}=\emptyset$ (ignore)
and 
(4) $H\left(p,V\left(p\right),z\right)\leq0$ for $z\in D^{-}$ where $D^{-}=\left[\frac{dU_{S}\left(p\right)}{dp}\mid_{p=\underline{p}},\frac{dV^{L}\left(p\right)}{dp}\mid_{p=\underline{p}}\right]$.

Step (1): 
When $p>\bar{p}$, we have $V\left(p\right)=U_{R}\left(p\right)$ and
$
H\left(p,U_{R}\left(p\right),U_{R}'\left(p\right)\right)
<0
$
if and only if $p>\bar{p}$.

Step (2):
It can be shown that $V^{L}\left(p\right)$ is convex.
At $p=\bar{p}$, we have 
$V^{L}\left(p\right)=U_{R}\left(p\right)$.
If we decrease $p$ by a little bit, $U_{R}\left(\cdot\right)$ decreases faster than $V^{L}\left(\cdot\right)$. 
Therefore, we have 
$V^{L}\left(p\right)\geq U_{R}\left(p\right)$ for $p\leq\bar{p}$. 
At 
$p=\underline{p}$, 
we have 
$V^{L}\left(\underline{p}\right)
=
U_{S}\left(\underline{p}\right)$ 
and 
$\frac{dV^{L}\left(p\right)}{dp}
>
\frac{dU_{S}\left(p\right)}{dp}$. 
As a result, we have 
$V^{L}\left(p\right)
\geq
U_{S}\left(p\right)$ 
for
$p>\underline{p}$.

Step (3):
If $p<\underline{p}$, then,
$V\left(p\right)=U_{S}\left(p\right)$.
To have $H\left(p,U_{S}\left(p\right),\frac{dU_{S}\left(p\right)}{dp}\right)\leq 0$,
we need
$p\leq \frac{\frac{c}{a}}{u_{H}-u^{S}}$.
It can be shown that $\underline{p}<\frac{\frac{c}{a}}{u_{H}-u^{S}}$.
As a result, if $p<\underline{p}$, then
$H\left(p,V\left(p\right),V'\left(p\right)\right)\leq 0$.

Step (4):
Since 
$H\left(p,V\left(p\right),z\right)$ is increasing in $z$ and we have $H\left(\underline{p},V\left(\underline{p}\right),\frac{dV^{L}\left(p\right)}{dp}\mid_{p=\underline{p}}\right)=0$, 
it is true that 
$H\left(p,V\left(p\right),z\right)\leq 0$
for
$
\forall z\in
\left[\frac{dU_{S}\left(p\right)}{dp}\mid_{p=\underline{p}},\frac{dV^{L}\left(p\right)}{dp}\mid_{p=\underline{p}}\right]
$.

To conclude,
if $c<\bar{c}$,
the DM's optimal strategy is to acquire the signal if the belief if in the range 
$\left(\underline{p},\bar{p}\right)$, 
to take action $R$ if $p\geq\bar{p}$ 
and to take action $S$ if $p\leq \underline{p}$.
If $c\geq\bar{c}$,
the DM's optimal strategy is to take an action without acquiring the signal.
\section{Proof of \Cref{lemma:rhocontinuous}}

Consider a strategy 
$
\gamma=
\left(
\rho, \sigma
\right)
$
as defined in \Cref{def:mixedstrategy}.
Suppose $\rho$ is discontinuous at some $\tau>0$ and continuous everywhere else.
Since by definition $\rho$ is right-continuous and weakly increasing, this implies that
$
\rho\left(\tau\right)>\lim_{t\rightarrow \tau_{-}}\rho\left(t\right).
$
%Suppose this mass that $\rho$ places on $\tau$ is 
%$
%M.
%$
Then, the induced probability that the player takes $R$ before or at time $t$ in state $\omega$, 
$
F_{\omega}^{\gamma}\left(t\right),
$ 
\footnote{The superscript indicates the strategy that this function is induced from.}
is right continuous such that
$
F_{\omega}^{\gamma}\left(\tau\right)
>
\lim_{t\rightarrow \tau_{-}}F_{\omega}^{\gamma}\left(t\right)
$
for 
$
\forall
\omega\in \{H,L\}
$.
Let 
$
M_{\omega}
$
denote the mass that 
$
F_{\omega}^{\gamma}\left(\tau\right)
$
places on $\tau$.
Consider a deviation 
$
\gamma^{D}=
\left(
\rho^{D}, \sigma
\right)
$ 
such that it is identical to $\gamma$ except that the mass that $\rho$ places on $\tau$ is shifted to $\tau-\triangle$.
There exists a $\triangle>0$ such that given the opponent uses the $\gamma$ strategy, using $\gamma^{D}$ gives the player a higher payoff than $\gamma$.

Suppose the opponent uses the $\gamma$ strategy described above.
Consider the history that the player receives no revealing signal until time $\tau-\triangle$ where $\triangle>0$.
The total gain from acquiring information for $\triangle$ time longer is
\begin{align*}
p_{\tau-\triangle}a\triangle 
\left[
u_{H}-F_{H}^{\gamma}\left(\tau\right)\bar{\triangle}_{H}
\right]
+
\left[
1-p_{\tau-\triangle}a\triangle-\left(1-p_{\tau-\triangle}\right)b\triangle \right]
U_{R}\left(\tau\right)
-
U_{R}\left(\tau-\triangle \right)
\end{align*}
where 
\begin{align*}
U_{R} \left(t\right)
=
\begin{cases}
p_{t}
\left[
u_{H}-F_{H}^{\gamma}\left(t\right)\bar{\triangle}_{H}
\right]
+
\left(1-p_{t}\right)
\left[
u_{L}-F_{L}^{\gamma}\left(t\right)\bar{\triangle}_{L}
\right]
&
t<\tau
\\
p_{t}
\left[
u_{H}-F_{H}^{\gamma}\left(t\right)\bar{\triangle}_{H}-\left(1-F_{H}^{\gamma}\left(t\right)\right)\underline{\triangle}_{H}
\right]
+
\left(1-p_{t}\right)
\left[
u_{L}-F_{L}^{\gamma}\left(t\right)\bar{\triangle}_{L}-\left(1-F_{L}^{\gamma}\left(t\right)\right)\underline{\triangle}_{L}
\right]
& 
t=\tau
\\
\end{cases}
.
\end{align*}
When $\triangle>0$ approaches zero, this gain approaches
$$
p_{\tau}
\left[
-M_{H}\bar{\triangle}_{H}-\left(1-F_{H}^{\gamma}\left(\tau\right)\right)\underline{\triangle}_{H}
\right]
+
\left(1-p_{\tau}\right)
\left[
-M_{L}\bar{\triangle}_{L}-\left(1-F_{L}^{\gamma}\left(\tau\right)\right)\underline{\triangle}_{L}
\right]
,
$$
which is negative due to the mass point $M_{\omega}$ and the tie-breaking at time $\tau$.
Therefore, there is always a deviation to put the mass places on $\tau$ to $\tau-\triangle$.
\section{Proof of \Cref{thm:mainresult}}
\label{appen:proofofmainresult}

Let 
$
\frac{p^{M}}{1-p^{M}}
:=
\frac{-u_{L}+\underline{\triangle}_{L}}{u_{H}-\underline{\triangle}_{H}}
$,
$
\frac{p^{L}}{1-p^{L}}
:=
\frac{-u_{L}}{u_{H}}
$
and 
$
\frac{\tilde{p}}{1-\tilde{p}}
:=\frac{-bu_{L}-c}{a\bar{\triangle}_{H}+c}
$
be three prior cutoffs.
As $u_{L}<0$, when $c<-bu_{L}$, all of the three cutoffs are positive.
The cutoff 
$
\underline{p}
$
is defined in \Cref{prop:singledmresult} and 
$
p^{*}
$
is a fixed point that will be defined later in step 5.

The method to find the symmetric equilibrium is `guess and verify'.
The following outlines the steps of the proof.
\begin{enumerate}
  \item I show that there exists an equilibrium where both players use the pure strategy 
  $
  \left(0,S\right)
  $
  when 
  $
  p_{0}\leq \underline{p}
  $.
  \item I show that there exists an equilibrium where both players use the pure strategy 
  $
  \left(0,R\right)
  $ 
  when 
  $
  p_{0}\geq p^{M}
  $.
  \item I show that there exists an equilibrium where both players use the Immediate Mix with No Learning Strategy when $p^{L}<p_{0}<p^{M}$.
  \item I show that there exists an equilibrium where both players use a Randomised Stopping Time Strategy when
  $
  p^{L} < p_{0} < \tilde{p}.
  $
  \item I show that there exists an equilibrium where both players use a Randomised Stopping Time Strategy
  when 
  $
  p^{*}<p_{0}<p^{L}.
  $
  \item I show that there exists an equilibrium where both players use the Mixed Learning Strategy when 
  $
  \underline{p}<p_{0}<p^{*}
  $.
\end{enumerate}

\paragraph{Step 1}

When $p_{0} \leq \underline{p}$,
\Cref{prop:singledmresult} implies that a single DM takes $S$ immediately.
If a player takes $S$ immediately, the other player is the single DM in the game.
Given the player takes $S$ immediately, the opponent's best response is to take $S$.
Therefore, when $p_{0} \leq \underline{p}$, both players taking $S$ immediately is an equilibrium.
This is summarised in the following lemma.

\begin{lemma}
If $p_{0}\leq\underline{p}$, then, there exists an equilibrium where both players use the strategy $\left(0,S\right)$.
\end{lemma}

\paragraph{Step 2}

I prove the following lemma.

\begin{lemma}
If $p_{0}\geq p^{M}$, there exists an equilibrium where both players use the pure strategy $\left(0,R\right)$
\end{lemma}

\begin{proof}
Suppose player $j$ uses the strategy $\left(0,R\right)$.
I check whether player $i$ wants to stop at time $0$ and take action $R$ or to acquire the signal for $dt$ longer. 
If player $i$ takes the immediate $R$ action at time $0$, the payoff is
$$
U_{R}^{IR}\left(0\right)
=
p_{0}\left(u_{H}-\underline{\triangle}_{H}\right)
+
\left(1-p_{0}\right)\left(u_{L}-\underline{\triangle}_{L}\right)
.
$$
If the player takes action $R$ at time $t>0$, 
the payoff is 
$$
U_{R}^{IR}\left(t\right)
=
p_{t}\left(u_{H}-\bar{\triangle}_{H}\right)+\left(1-p_{t}\right)\left(u_{L}-\bar{\triangle}_{L}\right)
$$
At time $0$, the gain from acquiring the signal for $dt$ longer is 
\begin{align}
\label{eq:netgainwhenTequalszero}
p_{0}adt\left(u_{H}-\bar{\triangle}_{H}\right)
+
\left[1-p_{0}adt-\left(1-p_{0}\right)bdt\right]U_{R}^{IR}\left(0+dt\right)
-
U_{R}^{IR}\left(0\right).
\end{align}
When $dt\rightarrow 0$, (\ref{eq:netgainwhenTequalszero}) tends to something negative, which is smaller than the cost of information.
Therefore, if the opponent stops at time $0$, player $i$ prefers taking action $R$ at time $0$ to acquiring the signal for $dt$ longer. 
At time $0$, player $i$ prefers action $R$ to action $S$ at time $0$ if 
$
U_{R}^{IR}\left(0\right)\geq 0.
$ 
Since 
$
p_{0}\geq p^{M}
$,
the inequality 
$
U_{R}^{IR}\left(0\right)\geq 0
$ 
holds.
\end{proof}

\paragraph{Step 3}

I show the following lemma.

\begin{lemma}
If $p^{L}<p_{0}<p^{M}$, then, 
there exists an equilibrium where both players use the Immediate Mix Strategy.
\end{lemma}

\begin{proof}

Suppose player $j$ uses the strategy $\left(\rho^{IM},\sigma^{IM}\right)$ such that 
$
\rho^{IM}\left(t\right)=m
$
for 
$
\forall t\in\left[0,\infty\right)
$
and 
$
\sigma^{IM}\left(t\right)=1-m
$
for
$\forall t\in\left[0,\infty\right)$.
Then,
player $i$'s payoff from taking action $R$ at time $0$ is
\begin{align*}
U_{R}^{IM}\left(0\right)
&= m\left[p_{0}u_{H}+\left(1-p_{0}\right)u_{L}\right]+\left(1-m\right)\left[p_{0}\left(u_{H}-\underline{\triangle}_{H}\right)+\left(1-p_{0}\right)\left(u_{L}-\underline{\triangle}_{L}\right)\right]
\\
&= p_{0}\left[m u_{H}+\left(1-m\right)\left(u_{H}-\underline{\triangle}_{H}\right)\right]+\left(1-p_{0}\right)\left[m u_{L}+\left(1-m\right)\left(u_{L}-\underline{\triangle}_{L}\right)\right].
\end{align*}
Player $i$'s payoff from taking action $R$ at time $t>0$ is
\begin{align*}
U_{R}^{IM}\left(t\right)
&= m\left[p_{t} u_{H}+\left(1-p_{t} \right)u_{L}\right]+\left(1-m\right)\left[p_{t} \left(u_{H}-\bar{\triangle}_{H}\right)+\left(1-p_{t} \right)\left(u_{L}-\bar{\triangle}_{L}\right)\right]
\\
&= 
p_{t}\left[u_{H}-\left(1-m\right)\bar{\triangle}_{H}\right]+\left(1-p_{t}\right)\left[u_{L}-\left(1-m\right)\bar{\triangle}_{L}\right].
\end{align*}
At time $t=0$, player $i$ does not want to acquire the information due to the same reasoning as in step 2.

Next I show that for the prior $p_{0}\in \left(p^{L},p^{M}\right)$, there exists $m\in\left(0,1\right)$ such that player $i$ is indifferent between taking action $R$ and action $S$ at time $t=0$.
The indifference requires 
$
0
=
U_{R}\left(0\right)
$.
That is,
\begin{align}
\label{eq:indifferenceconditionattime0}
\frac{p_{0}}{1-p_{0}}
& = -
\frac{  \left[mu_{L}+\left(1-m\right)\left(u_{L}-\underline{\triangle}_{L}\right)\right]}{\left[mu_{H}+\left(1-m\right)\left(u_{H}-\underline{\triangle}_{H}\right)\right]  }
.
\end{align}
For any $p_{0} \in \left(p^{L},p^{M}\right)$, there exists $m \in \left(0,1\right)$ such that (\ref{eq:indifferenceconditionattime0}) holds.
\end{proof}

\paragraph{Step 4}

This step shows the existence of the equilibrium where players use the randomised stopping time strategy.
I show that there exists an equilibrium where both players take action $R$ at each time $t \geq 0$ with positive rate.
Suppose cost $c$ is sufficiently small such that $p^{L}<\tilde{p}$.
Let
$
L_{t} =\frac{p_{t} }{1-p_{t} }
$ 
be the likelihood ratio.
I prove the following lemma.

\begin{lemma}
\label{lemma:randomisedstoppingequilibriumcondition}
Suppose 
$
c
$
is sufficiently small
and 
$
b<2a.
$
If
$
p^{L}
<
p_{0}
<
\tilde{p},
$
then, 
there exists an equilibrium in mixed strategies 
$
\left(\rho, \sigma \right)
$
such that 
$
\sigma\left(t\right)
=
0
$
for 
$
\forall t \in \mathbb{R}_{+}
$
and
$
\rho:\mathbb{R}_{+}\rightarrow\left[0,1\right]
$
satisfies the following conditions:
\begin{enumerate}
    \item \textbf{(Initial condition)} $\rho\left(0\right)=0$.
    \item \textbf{(Increasing condition)} 
    There exists a 
    $
    T>0
    $ 
    such that for $t\in\left[0,T\right]$, $\rho$ is a solution to the following differential equation
    \begin{align}
    \label{eq:rhodifferentialequation}
\frac{d\rho\left(t\right)}{dt}
&=
\frac{b\left(-u_{L}\right)-c-L_{t} c+b\bar{\triangle}_{L}\intop_{0}^{t}e^{-bs}\frac{d\rho\left(s\right)}{ds}ds-L_{t} \bar{\triangle}_{H}ae^{-at}\left(1-\rho\left(t\right)\right)}{L_{t} \bar{\triangle}_{H}e^{-at}+\bar{\triangle}_{L}e^{-bt}}       
    \end{align}
    with the initial condition $\rho \left(0\right)=0$
    where $\frac{d\rho}{dt}>0$ for $\forall t\in \left[0,T\right]$
    and $\rho \left(T\right)=1$.
    \item \textbf{(Terminal condition)}
    $\rho\left(t\right)=1$
    for 
    $
    \forall t>T.
    $
\end{enumerate}
\end{lemma}
In words, the equilibrium strategy is that the players stop and take action $R$ at each time with positive rate.

\begin{proof}
Fix player $j$'s strategy $\left(\rho,\sigma\right)$,
I first show that at each time $t$, player $i$ is indifferent between taking action $R$ now or $dt$ later if $\rho$ satisfies (\ref{eq:rhodifferentialequation}).
Then \Cref{lemma:existenceofrho} shows that when $c$ is sufficiently small and $b<2a$, if 
$
p_{0}
<
\tilde{p}
$,
we can find a $\rho$ function such that $\frac{d\rho}{dt}>0$.
Last, I show that if 
$
p_{0}
>
p^{L}
$, 
player $i$ prefers acquiring information to taking action $S$.

Suppose player $j$ uses the mixed strategy $\left(\rho,\sigma\right)$.
This strategy induces
\begin{align*}
F_{H} \left(t\right) 
& 
=
\intop_{0}^{t}\left[e^{-as}\left(1-\rho  \left(s\right)\right)\left(a+\frac{\frac{d\rho  \left(s\right)}{ds}}{1-\rho  \left(s\right)}\right)\right]ds
\\
& 
=
1-e^{-at}\left(1-\rho  \left(t\right)\right)
\end{align*}
and
\begin{align*}
F_{L} \left(t\right) 
&
=\intop_{0}^{t}\left[e^{-bs}\left(1-\rho  \left(s\right)\right)\frac{\frac{d\rho  \left(s\right)}{ds}}{1-\rho  \left(s\right)}\right]ds
\\
&
=\intop_{0}^{t}\left[e^{-bs}\frac{d\rho  \left(s\right)}{ds}\right]ds.
\end{align*}
Player $i$'s payoff from taking action $R$ at time $t$ is 
$$
U_{R} \left(t\right)
:=
p_{t} \left[u_{H}- F_{H} \left(t\right)\bar{\triangle}_{H}\right]
+
\left(1-p_{t} \right)\left[u_{L}- F_{L} \left(t\right)\bar{\triangle}_{L}\right].
$$
The equilibrium condition is that given player $j$'s strategy, player $i$ is indifferent between taking action $R$ and acquiring the signal at each time instant.
That is, 
\begin{align*}
p_{t} adt\left[u_{H}- F_{H} \left(t+dt\right)\bar{\triangle}_{H}\right]+\left[1-p_{t} adt-\left(1-p_{t} \right)bdt\right]U_{R}\left(t+dt\right)-U_{R}\left(t\right)
=
cdt.
\end{align*}
The interpretation is that the marginal cost and marginal benefit associated with acquiring information for $dt$ longer are the same.
When $dt\rightarrow 0$, we have the following differential equation
\begin{align*}
c
&=
\left(1-p_{t} \right)b\left[-u_{L}+ F_{L} \left(t\right)\bar{\triangle}_{L}\right]-p_{t} \frac{d F_{H} \left(t\right)}{dt}\bar{\triangle}_{H}-\left(1-p_{t} \right)\frac{d F_{L} \left(t\right)}{dt}\bar{\triangle}_{L}
\end{align*}
After plugging in the expressions of
$
F_{H}\left(t\right)
$
and
$
F_{L}\left(t\right),
$
we have
\begin{align*}
\frac{d\rho\left(t\right)}{dt}
&=
\frac{b\left(-u_{L}\right)-c-L_{t} c+b\bar{\triangle}_{L}\intop_{0}^{t}e^{-bs}\frac{d\rho\left(s\right)}{ds}ds-L_{t} \bar{\triangle}_{H}ae^{-at}\left(1-\rho\left(t\right)\right)}{L_{t} \bar{\triangle}_{H}e^{-at}+\bar{\triangle}_{L}e^{-bt}}.
\end{align*}

\begin{lemma}
\label{lemma:existenceofrho}
The differential equation (\ref{eq:rhodifferentialequation}) with initial condition 
$
\rho\left(0\right)=0
$
has a unique solution defined for all 
$
t\in\left[0, \infty\right)
$.
\end{lemma}

\begin{proof}
Let 
$
A\left(t\right)
:=
e^{-bt}\rho\left(t\right)
$. 
Then, we have $\frac{dA}{dt}=-bA+e^{-bt}\frac{d\rho}{dt}$. 
Use the formula that $L_{t}=L_{0}e^{\left(b-a\right)t}$, (\ref{eq:rhodifferentialequation}) can be rewritten as 
\begin{align}
\left(\frac{dA}{dt}+bA\right)e^{bt} 
& 
=
\frac{b\left(-u_{L}\right)-c-cL_{0}e^{\left(b-a\right)t}+b\bar{\triangle}_{L}\left(A+b\intop_{0}^{t}A\left(s\right)ds\right)-L_{0}\bar{\triangle}_{H}ae^{\left(b-2a\right)t}\left(1-e^{bt}A\right)}{L_{0}\bar{\triangle}_{H}e^{\left(b-2a\right)t}+\bar{\triangle}_{L}e^{-bt}}.
\end{align}
Let 
$
z\left(t\right)
:=
\intop_{0}^{t}A\left(s\right)ds
$. 
Then, the differential equation can be written as 
\begin{align}
\label{eq:zdifferentialequation}
z''+g_{1}\left(t\right)z'+g_{2}\left(t\right)z=g_{3}\left(t\right)
\end{align}
where 
\begin{align*}
g_{0}\left(t\right)
& =
\frac{1}{L_{0}\bar{\triangle}_{H}e^{2\left(b-a\right)t}+\bar{\triangle}_{H}},
\\
g_{1}\left(t\right)	
&=
b-\frac{b\bar{\triangle}_{L}+L_{0}\bar{\triangle}_{H}ae^{2\left(b-a\right)t}}{g_{0}\left(t\right)},
\\
g_{2}\left(t\right)	
&= 
-\frac{b^{2}\bar{\triangle}_{L}}{g_{0}\left(t\right)},
\\
g_{3}\left(t\right)	
&=
\frac{b\left(-u_{L}\right)-c-cL_{0}e^{\left(b-a\right)t}-L_{0}\bar{\triangle}_{H}ae^{\left(b-2a\right)t}}{g_{0}\left(t\right)}.
\end{align*}
Since 
$
g_{0}\left(t\right)
$,
$
g_{1}\left(t\right)
$, 
$
g_{2}\left(t\right)$
and
$g_{3}\left(t\right)$
are continuous on 
$\left[0,\infty\right)$, 
(\ref{eq:zdifferentialequation}) with initial conditions 
$z\left(0\right)=0$ 
and
$z'\left(0\right)=0$
has a unique solution defined for all $t$ on $\left[0,\infty\right)$. 
The existence of $z\left(\cdot\right)$ implies the existence of $\rho\left(\cdot\right)$.
\end{proof}

\begin{lemma}
$
\frac{d\rho}{dt}\mid_{t=0}
>
0
$
iff
$
p_{0}<\tilde{p}.
$
\end{lemma}

\begin{proof}
\begin{align*}
\frac{d\rho}{dt}\mid_{t=0}
&=
\frac{b\left(-u_{L}\right)-c-L_{0}c-L_{0}\bar{\triangle}_{H}a}{L_{0}\bar{\triangle}_{H}+\bar{\triangle}_{L}}
\\
&>0
\end{align*}
if and only if
$$
L_{0}=
\frac{p_{0}}{1-p_{0}}
<
\frac{b\left(-u_{L}\right)-c}{a\bar{\triangle}_{H}+c}
=
\frac{\tilde{p}}{1-\tilde{p}}.
$$
\end{proof}

This lemma implies that if
$
p_{0}<\tilde{p},
$
then,
$
\frac{d\rho}{dt}
>
0
$
at the neighbourhood of $0$.
The following lemma and its proof shows that under some conditions, 
if $\rho\left(\cdot\right)$ increases at the neighbourhood of $0$ and the value of $\rho$ is smaller than 1, 
then,
it continues increasing.

\begin{lemma}
\label{lemma:rhoincreasing}
Suppose 
$b<2a$ and $c$ sufficiently small. 
If 
$
\rho\left(t\right) \in \left[0,1\right)
$
for 
$
\forall t \in \left[0,\tau\right]
$
and 
$
\frac{d\rho}{dt}>0
$
$
\forall t \in \left[0,\tau\right),
$
then,
$
\frac{d\rho}{dt}>0
$
at $t=\tau$.
\end{lemma}

\begin{proof}
If $b<2a$, $c$ sufficiently small, and 
$
\frac{d\rho}{dt}>0
$
$
\forall t \in \left[0,\tau\right),
$
then,
\begin{align*}
    \Phi\left(t\right)
    :=
    \frac{b\left(-u_{L}+\bar{\triangle}_{L}\intop_{0}^{t}e^{-bs}\frac{d\rho}{ds}ds\right)-c-L_{0}e^{\left(b-a\right)t}c}
    {\bar{\triangle}_{H}ae^{\left(b-2a\right)t}\left(1-\rho\left(t\right)\right)}
\end{align*}
increases in $t$ for 
$
\forall t \in \left[0,\tau\right).
$
Since 
$\rho\left(t\right) \in \left[0,1\right)$
for
$
\forall t \in \left[0,\tau\right),
$
then,
$
\Phi\left(t\right)>L_0
$ 
for 
$
\forall t \in \left[0,\tau\right)
$
iff
$
\frac{d\rho}{dt}>0
$
for
$
\forall t \in \left[0,\tau\right).
$
Next, I show by contradiction that if 
$
\frac{d\rho}{dt}>0
$
for
$
\forall t \in \left[0,\tau\right)
$
and
$
\rho\left(t\right) \in \left[0,1\right)
$
for 
$
\forall t \in \left[0,\tau\right],
$
then, it must be that 
$
\frac{d\rho}{dt}>0
$
at $t=\tau$.
Suppose 
$
\frac{d\rho}{dt}=0
$
at $t=\tau$.
Then, it must be that 
$
\Phi\left(\tau\right)=L_0.
$ 
However, we know that 
$\Phi\left(t\right)$ 
increases in $t$ for 
$\forall t \in \left[0,\tau\right)$ and 
$
\Phi\left(t\right)>L_0
$ 
for 
$\forall t \in \left[0,\tau\right)$.
Then, it cannot be that 
$
\Phi\left(\tau\right)=L_0.
$ 
There is a contradiction and hence 
$
\frac{d\rho}{dt}\neq0
$
at $t=\tau$.
Since $\rho$ is continuous, it cannot be that
$
\frac{d\rho}{dt}>0
$
at $t=\tau$.
\end{proof}

Next I show that given \Cref{lemma:existenceofrho} and \Cref{lemma:rhoincreasing},
$
\rho\left(t\right)=1
$
will be reached in a finite time.

\begin{lemma}
Given \Cref{lemma:existenceofrho} and \Cref{lemma:rhoincreasing},
there exists a 
$
T<\infty
$ such that 
$
\rho\left(T\right)=1.
$
\end{lemma}

\begin{proof}
If 
$
\frac{d\rho}{dt}>0
$
and $\rho\left(t\right)\in\left[0,1\right)$,
then
\begin{align*}
    \Phi\left(t\right) 
    & <
    \frac{b\left(-u_{L}+\bar{\triangle}_{L}\rho\left(t\right)\right)-c-L_{0}e^{\left(b-a\right)t}c}{\bar{\triangle}_{H}ae^{\left(b-2a\right)t}\left(1-\rho\left(t\right)\right)}
    \\
    & <
    \frac{b\left(-u_{L}+\bar{\triangle}_{L}\right)-c-L_{0}e^{\left(b-a\right)t}c}{\bar{\triangle}_{H}ae^{\left(b-2a\right)t}\left(1-\rho\left(t\right)\right)}.
\end{align*}
As $t\rightarrow \infty$, we have $\Phi\left(t\right)<0$.
If $\Phi\left(t\right)<0$ and $\rho\left(t\right)<1$ when $t\rightarrow \infty$,
then,
$\frac{d\rho}{dt}<0$ as $t\rightarrow \infty$.
There is a contradiction.
Therefore, it must be that $\rho\left(t\right)>1$ for some $t$. Because of continuity, there exists a $T$ such that $\rho\left(T\right)=1$.
\end{proof}

\begin{lemma}
If 
$
p^{L}
<
p_{0}
$
and $\rho\left(t\right)$ satisfies the conditions in \Cref{lemma:randomisedstoppingequilibriumcondition},
then, $\sigma\left(t\right)=0$
for 
$\forall t\in \left[0,T\right]$.
\end{lemma}

\begin{proof}
Let $V\left(t\right)$ be the value associated with the mixed strategy characterised in \Cref{lemma:randomisedstoppingequilibriumcondition}.
Then, 
$$
V\left(t\right)
=
U_{R}\left(t\right)-ct
=
U_{R}\left(0\right)
$$
because the player's payoffs are the same at each time when she is randomising between continuing and stopping (to take action $R$).
Since 
$
p^{L}
<
p_{0},
$
we have 
$U_{R}\left(0\right)>0$.
Hence, at each $t<T$, we have $U_{R}\left(t\right)>0$.
As a result, $\sigma\left(t\right)=0$ for $\forall t \in \left[0,T\right]$.
\end{proof}

This completes the proof of \Cref{lemma:randomisedstoppingequilibriumcondition}.

\end{proof}

\paragraph{Step 5}

In the previous step, 
I consider the situation that the player starts randomisation at time $t=0$. 
However, it is possible that the player strictly prefers to acquire information for a certain time period and then starts randomisation.
This step shows the existence of the equilibrium where players use a mixed strategy
$\left(\hat{\rho},\hat{\sigma}\right)$ such that
(1) 
$\hat{\rho}\left(t\right)=0$ for $t \leq \hat{T}$,
(2) 
$
\frac{d\hat{\rho}\left(t\right)}{dt}\mid_{t\geq \hat{T}}>0,
$
(3)
$
\hat{T} \in \left[\hat{T}_{l}, \hat{T}_{r}\right]
$
and
(4)
$\hat{\sigma}\left(t\right)=0$ for $\forall t\geq 0$.
That is, a randomised stopping time strategy such that the players acquire information with probability one until some time $\hat{T}>0$ and then start randomising between stopping and acquiring information.
I first define three parameters
$
T_{l},
$
$
T_{r},
$
$
p^{*}
$
and show their existence.
Then, I show the conditions that $\hat{\rho}$ satisfies in equilibrium and under what conditions such equilibrium exists.

Let 
\begin{align}
\label{eq:defineTr}
T_{r}
:=
\min 
\left\{
t:
\frac{p_{0}}{1-p_{0}}
  e^{\left(b-a\right)t}
=
\frac{b\left(-u_{L}\right)-c}{c+ae^{-at}\bar{\triangle}_{H}}
\right\}
\end{align}
and
\begin{align}
\label{eq:defineTl}
T_{l}
:=
\min 
\left\{
t:
\frac{p_{0}}{1-p_{0}}
  e^{\left(b-a\right)t} 
 =
\frac{-u_{L}}{u_{H}-\left(1-e^{-at}\right)\bar{\triangle}_{H}}
\right\}
.
\end{align}

\begin{lemma}
\label{lemma:TrandTldefinitionandexistence}
If $p_{0}<\tilde{p}$,
then, $T_{r}>0$ exists.
If 
$
p_{0}< p^{L}
$,
then,
$T_{l}>0$ 
exists.
\end{lemma}

Let
\begin{align}
\label{eq:pstar}
\frac{p^{*}\left(t\right)}{1-p^{*}\left(t\right)}
:=
\frac{\frac{c}{b}}{u_{H}-\frac{1}{2}\bar{\triangle}_{H}-\frac{c}{a}+J_{2}\left(t\right)}
\end{align}
where
$$
J_{2}\left(t\right)
=
\left[\frac{1}{2}e^{-at}\bar{\triangle}_{H}+\frac{c}{a}-\left(-u_{L}-\frac{c}{b}\right)\frac{1}{L_{t}}\right]e^{-at}
.
$$
Let
$$
\underline{p}^{*}\left(p_{0}\right)
:=
p^{*}\left(T_{r}\right)
$$
be
$
p^{*}\left(t\right)
$ 
evaluated at $t=T_{r}$.
I denote it as $\underline{p}^{*}\left(p_{0}\right)$ because $T_{r}$ depends on $p_{0}$.
Let $p^{*}$ be the fixed point such that 
$
p^{*}
=
\underline{p}^{*}\left(p^{*}\right).
$

\begin{lemma}
\label{lemma:pstarexistence}
Suppose $c$ is sufficiently small.
There exists a
$
p^{*}
<
p^{L}
$
such that 
$
p^{*}
=
\underline{p}^{*}\left(p^{*}\right).
$
We have
$
p^{*}<p_{0} 
$
if and only if
$
\underline{p}^{*}\left(p_{0}\right)
<
p_{0}.
$
\end{lemma}

\begin{proof}

When
$
0<p_{0}<p^{L},
$
$
\underline{p}^{*}\left(\cdot\right)
$
decreases in 
$
p_{0}
$.
When 
$
p_{0} \rightarrow 0,
$
$
p_{0}<\underline{p}^{*}\left(p_{0}\right)<p^{L}.
$
When 
$
p_{0} \rightarrow p^{L}
$
and $c$ sufficiently small, 
we have
$
\underline{p}^{*}\left(p_{L}\right)<\underline{p}^{*}\left(0\right)< p^{L}.
$
Therefore, there exists $p^{*}<p^{L}$ such that 
$
p^{*}
=
\underline{p}^{*}\left(p^{*}\right)
$.
Since 
$
\underline{p}^{*}\left(\cdot\right)
$
decreases in $p_{0}$,
$
p^{*}<p_{0} 
$
if and only if
$
\underline{p}^{*}\left(p_{0}\right)
<
p_{0}.
$
\end{proof}

Next, I show the existence of the equilibrium and the conditions $\hat{\rho}$ satisfies in equilibrium.
Suppose player $j$ uses the mixed strategy 
$\left(\hat{\rho},\hat{\sigma}\right)$ such that
(1) 
$\hat{\rho}\left(t\right)=0$ for $t \leq \hat{T}$,
(2) 
$
\frac{d\hat{\rho}\left(t\right)}{dt}\mid_{t\geq \hat{T}}>0,
$
(3)
$
\hat{T} \in \left[\hat{T}_{l}, \hat{T}_{r}\right]
$
and
(4)
$\hat{\sigma}\left(t\right)=0$ for $\forall t\geq 0$.
Given the assumption that $\hat{\rho}\left(t\right)=0$ for $t \leq \hat{T}$ and 
$\hat{\rho}\left(t\right)>0$ for $t > \hat{T}$,
consider time $\hat{T}$ as a new time $0$ and denote the time line starting from $\hat{T}$ as $\tau$.
Let $\tau:=t-\hat{T}$ and
$\lambda\left(\tau\right):=\hat{\rho}\left(\tau+\hat{T}\right)$.
Let
$q_{\tau}=p_{\tau+\hat{T}}$
be the belief
and 
let $\hat{U}_{R}\left(\tau\right)$ be
player $i$'s payoff from taking action $R$ at time $\tau$, then,
\begin{align*}
\hat{U}_{R}\left(\tau\right)
=
&
q_{\tau} \left[\left(1-\hat{F}_{H} \left(\tau\right)\right)\eta+\hat{F}_{H} \left(\tau\right)\left(u_{H}-\bar{\triangle}_{H}\right)\right]+\left(1-q_{\tau} \right)\left[u_{L}-\hat{F}_{L} \left(\tau\right)\bar{\triangle}_{L}\right]
\end{align*}
where
$$
\hat{F}_{H} \left(\tau\right)
=
1-e^{-a\tau-a\hat{T}}+e^{-a\tau-a\hat{T}}\lambda\left(\tau\right),
$$
$$
\hat{F}_{L} \left(\tau\right)
=
\int_{0}^{\tau+\hat{T}}e^{-bs}\lambda\left(s\right)ds
$$
and
$
\eta
:=
u_{H}-\left(1-e^{-a\hat{T}}\right)\bar{\triangle}_{H}.
$
The intuition is that when considering time $\hat{T}$ as a new time $0$, player $i$'s problem is essentially the same as in step 4 with $\eta$ being the payoff for the first action $R$ taker instead of $u_{H}$.

After time $\hat{T}$, the equilibrium condition requires player $i$ being indifferent between acquiring the information and taking action $R$ at each time instant $\tau>0$.
That is,
\begin{align*}
c
=&
\left(1-q_{\tau} \right)b\left[-\left(u_{L}-\hat{F}_{L}\left(\tau\right)\bar{\triangle}_{L}\right)\right]
\\
& 
-q_{t} \frac{d\hat{F}_{L}\left(\tau\right)}{d\tau}\left(\eta-\left(u_{H}-\bar{\triangle}_{H}\right)\right)-\left(1-q_{\tau} \right)\frac{d\hat{F}_{L}\left(\tau\right)}{d\tau}\bar{\triangle}_{L}.
\end{align*}

After plugging in the expressions of 
$
\hat{F}_{H}\left(t\right)
$
and 
$
\hat{F}_{H}\left(t\right),
$
we have
\begin{align}
\label{eq:lambdadifferentialequation}    
\frac{d\lambda\left(\tau\right)}{d\tau}
=&
\frac{b\left(-u_{L}\right)-c-\hat{L}_{\tau} c+b\bar{\triangle}_{L}\intop_{0}^{\tau}e^{-bs}\frac{d\lambda\left(s\right)}{ds}ds-\hat{L}_{\tau} \bar{\triangle}_{H}ae^{-a\tau}\left(1-\lambda\left(t\right)\right)}{\hat{L}_{\tau} \left(\eta-\left(u_{H}-\bar{\triangle}_{H}\right)\right)e^{-a\tau}+\bar{\triangle}_{L}e^{-b\tau}}
\end{align}
where $\hat{L}_{\tau} =\frac{q_{\tau} }{1-q_{\tau} }$.
To have $\frac{d\lambda\left(s\right)}{ds}>0$, we need
\begin{align*}
\hat{L}_{\tau} 
& <
\frac{b\left(-u_{L}\right)-c+b\bar{\triangle}_{L}\intop_{0}^{\tau}e^{-bs}\frac{d\lambda\left(s\right)}{ds}ds}{c+\left(\eta-\left(u_{H}-\bar{\triangle}_{H}\right)\right)ae^{-a\tau}\left(1-\lambda\left(\tau\right)\right)}.
\end{align*}
That is,
\begin{align}
\label{eq:conditionforrhohat}
L_{0}e^{\left(b-a\right)t} 
& <
\frac{b\left(-u_{L}\right)-c+b\bar{\triangle}_{L}\intop_{0}^{t}e^{-bs}\frac{d\hat{\rho}\left(s\right)}{ds}ds}{c+\bar{\triangle}_{H}ae^{-at}\left(1-\hat{\rho}\left(t\right)\right)}
\end{align}
for all $t>0$.
Equation (\ref{eq:conditionforrhohat}) is derived by substituting in $\tau=t-\hat{T}$, 
$\lambda\left(\tau\right)=\hat{\rho}\left(t\right)$ and 
$
\eta
=
u_{H}-\left(1-e^{-a\hat{T}}\right)\bar{\triangle}_{H}
$.
The existence of an increasing function $\hat{\rho}\left(\cdot\right)$ has been shown in  \Cref{lemma:existenceofrho}.
that satisfies (\ref{eq:conditionforrhohat}).

Next, I characterise the condition that $\hat{T}$ satisfies in equilibrium.
I am going to show that there exists an interval 
$
\left[\hat{T}_{l},\hat{T}_{r}\right]
$
such that an equilibrium exists when $\hat{T} \in \left[\hat{T}_{l},\hat{T}_{r}\right]$.
The idea is that before time $\hat{T}$, player $i$ must strictly prefer to acquire the signal and after time $\hat{T}$, player $i$ is indifferent between acquiring the signal and taking action $R$ at each time instant.
To have the player strictly prefer to acquire the signal before time $\hat{T}$, we need the marginal cost smaller than the marginal benefit associated with acquiring the signal.
That is,
\begin{align}
\label{eq:marginalcost=marginalbenefit}
L_{0} e^{\left(b-a\right)t}
<
\frac{b\left(-u_{L}\right)-c}{c+ae^{-at}\bar{\triangle}_{H}}
\end{align}
for all $t\leq \hat{T}$.
The upperbound $\hat{T}_{r}$ is the first time the marginal cost of acquiring the signal exceeds the marginal benefit.
Given \Cref{lemma:TrandTldefinitionandexistence}, 
if $p_{0}<\tilde{p}$, then,
$
\hat{T}_{r}=T_{r}>0
$
exists.

The lowerbound of $\hat{T}$ is the earliest time point at which the player is willing to start randomising.
That is,
if player $j$ starts randomising at time $\hat{T}$, player $i$ must prefer to start randomising at time $\hat{T}$ instead of taking action $S$.
%if the payoff from randomising is greater than $ $.
At time $\hat{T}$, the value associated with randomisation is the same as the value associated with taking action $R$ because of the opponent's randomisation.
Therefore, in order to have player $i$ prefer randomisation to taking action $S$ at time $\hat{T}$, we need
\begin{align*}
p_{\hat{T}}\left[u_{H}-\left(1-e^{-a\hat{T}}\right)\bar{\triangle}_{H}\right]+\left(1-p_{\hat{T}}\right)u_{L}
& \geq 
0.
\end{align*}
That is,
\begin{align}
\label{eq:conditionforlowerboundThat}
L_0 e^{\left(b-a\right)\hat{T}} 
& \geq
\frac{-u_{L}}{u_{H}-\left(1-e^{-a\hat{T}}\right)\bar{\triangle}_{H}}.
\end{align}
The lowerbound of $\hat{T}$ is the smallest $\hat{T}$ such that inequality (\ref{eq:conditionforlowerboundThat}) holds.
Given \Cref{lemma:TrandTldefinitionandexistence}, 
if $p_{0}< p^{L}$, then,
$
\hat{T}_{l}
=T_{l}
>0.
$
If 
$p_{0}\geq p^{L}$, then,
$
\hat{T}_{l}
=0.
$

The following lemma characterises the equilibrium when $p^{*}<p_{0}<\tilde{p}$.

\begin{lemma}
\label{lemma:RSTequilibriumexistence}
Suppose $b<2a$ and $c$ is sufficiently small.
If $p^{*}<p_{0}<\tilde{p}$,
there exists an equilibrium in mixed strategies
$
\left(\hat{\rho},\hat{\sigma}\right)
$
such that
$\hat{\sigma}\left(t\right)=0$
for $\forall t \in \mathbb{R}_{+}$
and $\hat{\rho}\left(\cdot\right)$ satisfies the following conditions:
\begin{enumerate}
    \item $\hat{\rho}\left(t\right)=0$ for $t \leq \hat{T}$
    \item $\hat{\rho}\left(t\right) \in \left( 0, 1 \right]$ for $t \in \left(\hat{T}, \bar{T}\right]$ 
    and 
    $\hat{\lambda}\left(\tau \right)=\hat{\rho}\left(\tau+\hat{T}\right)$ is a solution to the differential equation (\ref{eq:lambdadifferentialequation})
    with initial condition 
    $\hat{\lambda}\left(0 \right)=0$ where 
    $\frac{d\hat{\rho}}{dt}>0$ for $\forall t \in \left[\hat{T}, \bar{T}\right]$
    \item $\hat{\rho}\left(t\right)=1$ for $\forall t>\bar{T}$
    \item $\hat{T} \in \left[\hat{T}_{l},\hat{T}_{r}\right]$ 
    where 
    $\hat{T}_{r}=T_{r}$, 
    $\hat{T}_{l}=T_{l}$ if $p_{0}< p^{L}$ 
    and
    $\hat{T}_{l}=0$ if $p_{0}\geq p^{L}$.
\end{enumerate}
\end{lemma}

\begin{proof}
I have shown the existence of an increasing $\hat{\rho}$ function, and (\ref{eq:lambdadifferentialequation}) guarantees that the players stop and take action $R$ with a positive rate at each time $t \in \left[\hat{T}, \bar{T}\right]$.
I have also derived the conditions for $\hat{T}$.
What left to show is the equilibrium exists when $p_{0}>p^{*}$.
Given \Cref{lemma:pstarexistence}, 
$
p^{*}<p_{0} 
$
if and only if
$
\underline{p}^{*}\left(p_{0}\right)
<
p_{0}.
$
What left to show is the equilibrium exists when $
\underline{p}^{*}\left(p_{0}\right)
<
p_{0}.
$.
Suppose player $j$ uses
$
\left(\hat{\rho},\hat{\sigma}\right)
$
strategy described in the lemma.
Since I have discussed what happens after time $\hat{T}$, I will characterise player $i$'s value $W\left(\cdot\right)$ at time $t<\hat{T}$.
The HJB equation is
\begin{align*}
\max
\left\{ 
p_{t} a\left[u_{H}-\left(1-e^{-at}\right)\bar{\triangle}_{H}-W\left(t\right)\right]
\right.
\\
+\left(1-p_{t} \right)b\left[-W\left(t\right)\right]-c+W'\left(t\right),
\\
\left.U\left(t\right)-W\left(t\right)\right\}  
& =
0
\end{align*}
If the learning region exists, in the learning region, we have
\begin{align*}
W\left(t\right)
= & 
p_{t} \left(\left(u_{H}-\bar{\triangle}_{H}\right)-\frac{c}{a}+\frac{1}{2}\bar{\triangle}_{H}e^{-at}\right)+\left(1-p_{t} \right)\left(-\frac{c}{b}\right) +p_{t} e^{at}J_{2}
\end{align*}
where $J_{2}$ is a constant. 
At time $\hat{T}$, 
we have 
$W\left(\hat{T}\right)
=
U\left(\hat{T}\right)=U_{R}\left(\hat{T}\right)$. 
This pins down $J_{2}$ as a function of $\hat{T}$, 
where 
$$
J_{2}\left(\hat{T}\right)
=
\left[\frac{1}{2}e^{-a\hat{T}}\bar{\triangle}_{H}+\frac{c}{a}-\left(-u_{L}-\frac{c}{b}\right)\frac{1}{L_{\hat{T}}}\right]e^{-a\hat{T}}
.$$
Player 1 acquires the signal at time $t=0$ if 
$W\left(0\right)\geq 0$, 
which requires 
\begin{align}
\label{eq:priorconditionacquiresignalattimezero}
L_0
\geq
\frac{\frac{c}{b}}{u_{H}-\frac{1}{2}\bar{\triangle}_{H}-\frac{c}{a}+J_{2}\left(\hat{T}\right)}
:=
\underline{L}\left(\hat{T}\right)
:=
\frac{p^{*}\left(\hat{T}\right)}{1-p^{*}\left(\hat{T}\right)}
.
\end{align}
The discussion shows that given the opponent uses the strategy described in the lemma,
player $i$'s best response is to use the Randomised Stopping Time Strategy if 
$
p_{0}\geq p^{*}\left(\hat{T}\right).
$

Since $\hat{T}$ can be any value in the interval $\left[\hat{T}_{l},\hat{T}_{r}\right]$, I use this to characterise the lowerbound of the prior $p_{0}$ such that the equilibrium described in \Cref{lemma:RSTequilibriumexistence} exists.

\begin{lemma}
\label{lemma:lowerbound_is_underline_p_star}
$
\underline{p}^{*}\left(p_{0}\right)
\leq
p^{*}\left(\hat{T}\right)
$
for 
$
\hat{T}\in\left[\hat{T}_{l},\hat{T}_{r}\right]
$.
\end{lemma}

\Cref{lemma:lowerbound_is_underline_p_star} is true because when $\hat{T}\leq\hat{T}_{r}$,
$J_{2}\left(\cdot\right)$ increases in its argument.
Since 
$
\frac{p^{*}\left(\hat{T}\right)}{1-p^{*}\left(\hat{T}\right)}
$
decreases in $J_{2}$,
given $\hat{T} \leq \hat{T}_{r}$,
we have 
$$
\underline{p}^{*}\left(p_{0}\right)
\leq
p^{*}\left(\hat{T}\right).
$$
\end{proof}

\Cref{lemma:lowerbound_is_underline_p_star} implies that for any $p_{0}>\underline{p}^{*}\left(p_{0}\right)$, there exists an equilibrium where the players use the Randomised Stopping Time Strategy as described in \Cref{lemma:RSTequilibriumexistence}.

\paragraph{Step 6}

Suppose player $j$ uses  the Mixed Learning Strategy 
$
\left(\rho^{ML},\sigma^{ML}\right)
$
such that
(1)
$\sigma^{ML}\left(t\right)=\beta>0$ for $\forall t \geq 0$,
(2)
$
\rho^{ML}\left(t\right) \in \left(0,1\right]
$ 
for 
$
t\in \left(\hat{T}^{\beta},\bar{T}^{\beta}\right],
$
(3)
$
\rho^{ML}\left(t\right)=1-\beta$ for $\forall t>\bar{T}^{\beta},
$
and
(4)
$
\hat{T}^{\beta} \in \left[\hat{T}_{l}^{\beta},\hat{T}_{r}^{\beta}\right]
$
where
$$
\hat{T}^\beta_{r}
:=
\min\left\{ t:
L_{0}e^{\left(b-a\right)t}
=
\frac{b\left(-u_{L}\right)-c}{c+\left(1-\beta\right)ae^{-at}\bar{\triangle}_{H}}
\right\}.
$$
and
$$
\hat{T}_{l}^{\beta}
:=
\min\left\{ t:L_{0}e^{\left(b-a\right)t}=\frac{-u_{L}}{e^{-aT}u_{H}-\left(1-\beta\right)\left(e^{-at}\right)\bar{\triangle}_{H}}\right\}.
$$
That is, she uses the Immediate action $S$ Strategy and the Randomised Stopping Time Strategy with probability $\beta \in \left(0,1\right)$ and $1-\beta$ such that $\hat{T}^{\beta}$ is the time at which player $j$ starts randomising conditional on she uses the Randomised Stopping Time Strategy.

Given the assumption that $\rho^{ML}\left(t\right)=0$ for $t \leq \hat{T}^{\beta}$ and 
$\rho^{ML}\left(t\right)>0$ for $t > \hat{T}^{\beta}$,
consider time $\hat{T}^{\beta}$ as a new time $0$ and denote the time line starting from $\hat{T}^{\beta}$ as $\tau$.
Let $\tau:=t-\hat{T}^{\beta}$
and
$\lambda^{ML}\left(\tau\right):=\rho^{ML}\left(\tau+\hat{T}\right)$.
Let
$
q_{\tau}^{ML}:=p_{\tau+\hat{T}}
$
be the belief and let
$
L_{\tau}^{ML}:=\frac{q_{\tau}^{ML}}{1-q_{\tau}^{ML}}
$
be the likelihood ratio.
Following the same discussion as in step 5, at time 
$
t\geq \hat{T}^{\beta},
$
that is,
$\tau\geq 0$,
player $i$ is indifferent between acquiring information and taking action $R$ if
\begin{align}
\label{eq:differentialequationlambda_ML}
\frac{d\lambda^{ML}\left(\tau\right)}{d\tau}
=&
\frac{b\left(-u_{L}\right)-c-L_{\tau}^{ML}c+b\bar{\triangle}_{L}\intop_{0}^{\tau}e^{-bs}\frac{d\lambda^{ML}\left(s\right)}{ds}ds-L_{\tau}^{ML}\bar{\triangle}_{H}ae^{-a\tau}\left(1-\lambda^{ML}\left(t\right)\right)}{L_{\tau}^{ML}\left(\eta^{ML}-\left(u_{H}-\bar{\triangle}_{H}\right)\right)e^{-a\tau}+\bar{\triangle}_{L}e^{-b\tau}}
\end{align}
where
$
\eta^{ML}:=u_{H}-\left(1-\beta\right)\left(1-e^{-a\hat{T}}\right)\bar{\triangle}_{H}.
$
The intuition is that when considering time $\hat{T}^{\beta}$ as a new time $0$, player $i$'s problem is essentially the same as in step 5 with $\eta^{ML}$ being the payoff for the first action $R$ taker instead of $\eta$.
The existence of an increasing function 
$
\rho^{ML}\left(\cdot\right)
$
can be shown following the same logic as in the proof of \Cref{lemma:existenceofrho}.

Next, I characterise the upperbound and lowerbound of $\hat{T}^{\beta}$.
Following similar argument as in step 5,
the upperbound $\hat{T}^{\beta}_{r}$ is the first time the marginal cost of acquiring information exceeds the marginal benefit.
At time $0\leq t<\hat{T}^{\beta}$, given player $j$'s strategy, player $i$'s payoff associated with taking action $R$ at time $t$ is
\begin{align*}
U_{R}^{\beta}\left(t\right)
= & 
\beta\left[p_{t} u_{H}+\left(1-p_{t} \right)u_{L}\right]
\\
& +
\left(1-\beta\right)\left[p_{t}\left(u_{H}-\left(1-e^{-at}\right)\bar{\triangle}_{H}\right)+\left(1-p_{t}\right)u_{L}\right].
\end{align*}
That is,
\begin{align*}
U_{R}^{\beta}\left(t\right)
= & 
p_{t} \left[\beta u_{H}+\left(1-\beta\right)\left(u_{H}-\left(1-e^{-at}\right)\bar{\triangle}_{H}\right)\right]+\left(1-p_{t} \right)u_{L}.
\end{align*}
The marginal cost of acquiring information is smaller than the marginal benefit if
\begin{align*}
\frac{b\left(-u_{L}\right)-c}{c+\left(1-\beta\right)ae^{-at}\bar{\triangle}_{H}}
& >L_{0} e^{\left(b-a\right)t}
%\\
%\frac{b\left( - u_{L} \right)-c}{ce^{\left(b-a\right)t}+\left(1-\beta\right)ae^{\left(b-2a\right)t}\left( u_{H} -\underline{u}_{H}^{R} \right)} 
%& \geq L_0
\end{align*}
Then, if 
$
L_{0}
<
\frac{b\left(-u_{L}\right)-c}{c+\left(1-\beta\right)a\bar{\triangle}_{H}}
$,
there exists a $\bar{T}^\beta_{r}>0$ such that
$$
\hat{T}^\beta_{r}
:=
\min\left\{ t:
L_{0}e^{\left(b-a\right)t}
=
\frac{b\left(-u_{L}\right)-c}{c+\left(1-\beta\right)ae^{-at}\bar{\triangle}_{H}}
\right\}.
$$
The lowerbound of $\hat{T}^{\beta}_{l}$ is the earliest time point at which the player is willing to start randomising instead of taking action $S$ given player $j$'s strategy.
Following similar argument as in step 5,
if $p_{0}<p^{L}$, then, 
there exists a $\bar{T}^\beta_{l}>0$ such that
$$
\hat{T}_{l}^{\beta}
:=
\min\left\{ t:L_{0}e^{\left(b-a\right)t}=\frac{-u_{L}}{e^{-aT}u_{H}-\left(1-\beta\right)\left(e^{-at}\right)\bar{\triangle}_{H}}\right\} .
$$

The following lemma characterises the equilibrium when $\underline{p}<p_{0}<p^{*}$.

tbc

\begin{lemma}
Suppose $b<2a$ and $c$ is sufficiently small.
If $\underline{p}<p_{0}<p^{*}$,
then, 
there exists an equilibrium where both players use the Mixed Learning Strategy 
$
\left(\rho^{ML},\sigma^{ML}\right)
$
such that
\begin{enumerate}
    \item $\sigma^{ML}\left(t\right)=\beta>0$ for $\forall t \geq 0$.
    \item $\rho^{ML}\left(t\right)=0$ for $t<\hat{T}^{\beta}$.
    \item 
    $
    \rho^{ML}\left(t\right) \in \left(0,1-\beta \right]
    $ 
    for 
    $
    t\in \left(\hat{T}^{\beta},\bar{T}^{\beta}\right]
    $
    and
    $\lambda^{ML}\left(\tau\right)=\rho^{ML}\left(\tau+\hat{T}\right)$
    is a solution to the differential equation (\ref{eq:differentialequationlambda_ML})
    with initial condition 
    $\lambda^{ML}\left(0 \right)=0$
    where 
    $
    \frac{d\rho^{ML}}{dt}>0
    $
    for $\forall t \in \left[ \hat{T}^{\beta},\bar{T}^{\beta} \right]$.
    \item 
    $
    \rho^{ML}\left(t\right)=1-\beta
    $ 
    for 
    $
    \forall t>\bar{T}^{\beta}.
    $
    \item 
    $
    \hat{T}^{\beta} \in \left[\hat{T}_{l}^{\beta},\hat{T}_{r}^{\beta}\right].
    $
\end{enumerate}
\end{lemma}

\begin{proof}
Suppose player $j$ uses  the Mixed Learning Strategy 
$
\left(\rho^{ML},\sigma^{ML}\right)
$.

I have shown that at time $t \geq \hat{T}^\beta$, player $i$ is indifferent between taking action $R$ and acquiring information for $dt$ longer
and that at time $0<t<\hat{T}^{\beta}$, player $i$ prefers acquiring information to taking action $R$.
What left is to show that if $\underline{p}<p_{0}<\underline{p}^{*}$, there exists a $\beta \in \left(0,1\right)$ such that player $i$ is indifferent between the Immediate action $S$ Strategy and the Randomised Stopping Time Strategy at time $0$.
Given \Cref{lemma:pstarexistence},
I show that such $\beta \in \left(0,1\right)$ exists if $\underline{p}<p_{0}<\underline{p}^{*}\left(p_{0}\right)$.

At time $t<\hat{T}^{\beta}$, player $i$'s value associated with acquiring information is
\begin{align*}
W^{\beta}\left(t\right)
=
&
p_{t}\left[u_{H}-\left(1-\beta\right)\bar{\triangle}_{H}+\frac{1}{2}\left(1-\beta\right)\bar{\triangle}_{H}e^{-at}-\frac{c}{a}\right]
\\
&
-\left(1-p_{t}\right)\frac{c}{b}+p_{t}e^{at}H\left(\hat{T}^{\beta}\right)
\end{align*}
where 
$$
H\left(\hat{T}^{\beta}\right)
=
\left[\frac{1}{2}\left(1-\beta\right)\bar{\triangle}_{H}e^{-a\hat{T}^{\beta}}+\frac{c}{a}-\frac{1}{L_{\hat{T}^{\beta}}}\left(-u_{L}-\frac{c}{b}\right)\right]e^{-a\hat{T}^{\beta}}
$$
is a constant.
The value of acquiring the information at time $t=0$ equals the value associated with taking action $S$ without acquiring the signal if 
$
W^{\beta}\left(0\right)
 =
0.
$
That is,
\begin{align*}
L_0
&=
\frac{\frac{c}{b}}{u_{H}-\frac{1}{2}\bar{\triangle}_{H}+\frac{1}{2}\beta\bar{\triangle}_{H}-\frac{c}{a}+H\left(\hat{T}^{\beta}\right)}
:=\underline{L}^{\beta}\left(\hat{T}^{\beta}\right).
\end{align*}
Then, 
$\beta$ can be pinned down by
$L_{0}=\underline{L}^{\beta}\left(\hat{T}^{\beta}\right)$.

\begin{lemma}
\label{lemma:existenceofbeta}
When 
$
\underline{p}
<
p_{0}
<
\underline{p}^{*}\left(p_{0}\right)
$, 
there exists a 
$
\beta\in\left(0,1\right)
$ 
such that 
$
L_{0}
=
\underline{L}^{\beta}\left(\hat{T}^{\beta}_{r}\right)
$. 
\end{lemma}

\begin{proof}
The proof uses the intermediate value theorem.
When $\beta=0$, we have 
$
\underline{L}^{\beta=0}\left(\hat{T}^{\beta=0}_{r}\right)
=
\underline{L}\left(\hat{T}_{r}\right)
.
$
Since 
$
p_{0}
<
\underline{p}^{*}\left(p_{0}\right),
$
we have 
$
\underline{L}^{\beta=0}\left(\hat{T}^{\beta=0}_{r}\right)
=
\underline{L}\left(\hat{T}_{r}\right)
>
L_{0}.
$

Next, I show that when $\beta=1$, we have 
$
\underline{L}^{\beta=1}\left(\hat{T}^{\beta=1}_{r}\right)
<
\underline{L}.
$
When
$\beta=1$, $\hat{T}^{\beta=1}_{r}$ satisfies
$$
L_{0}e^{\left(b-a\right)\hat{T}^{\beta=1}_{r}}
=
\frac{b\left( - u_{L} \right)-c}{c}
=
\bar{L}.
$$
Then,
$$
\underline{L}^{\beta=1}\left(\hat{T}_{r}^{\beta=1}\right)
=
\frac{\frac{c}{b}}{u_{H}-\frac{c}{a}+H\left(\hat{T}_{r}^{\beta=1}\right)}
$$
where 
$
H\left(\hat{T}_{r}^{\beta=1}\right)
=
\left(\frac{c}{a}-\frac{c}{b}\right)e^{-a\hat{T}_{r}^{\beta}}.
$
Then, we have
$$
\underline{L}^{\beta=1}\left(\hat{T}_{r}^{\beta=1}\right)\left[ u_{H} -\frac{c}{a} \right]+\left(\frac{c}{a}-\frac{c}{b}\right)e^{-a\hat{T}_{r}^{\beta}}\underline{L}^{\beta=1}\left(\hat{T}_{r}^{\beta=1}\right)
=
\frac{c}{b}.
$$
Since 
$
L_{0}e^{\left(b-a\right)\hat{T}^{\beta=1}_{r}}
=
\bar{L},
$
from the proof of \cref{prop:singledmresult} we know that 
$$
\underline{L}\left[ u_{H} -\frac{c}{a} \right]+\left(\frac{c}{a}-\frac{c}{b}\right)e^{-a\hat{T}_{r}^{\beta}}\left(\frac{\underline{L}}{L_{0}}\right)^{\frac{a}{b-a}}\underline{L}
=
\frac{c}{b}.
$$
Since $\underline{p}<p_{0}$,
we have 
$
\underline{L}^{\beta=1}\left(\hat{T}^{\beta=1}_{r}\right)
<
\underline{L}
$
and thus
$
\underline{L}^{\beta=1}\left(\hat{T}^{\beta=1}_{r}\right)
<
L_{0}.
$

Since 
$
\underline{L}^{\beta=1}\left(\hat{T}^{\beta=1}_{r}\right)
<
L_{0}
<
\underline{L}^{\beta=0}\left(\hat{T}^{\beta=0}_{r}\right)
$
there exists a $\beta \in \left(0,1\right)$ such that $L_{0}=\underline{L}^{\beta}\left(\hat{T}^{\beta}_{r}\right)$.
\end{proof}

Given \Cref{lemma:existenceofbeta}, we know that when 
$
\underline{p}
<
p_{0}
<
\underline{p}^{*}\left(p_{0}\right)
$,
there exists at least a pair $(\beta \in (0,1), \hat{T}_{r}^{\beta})$ such that the players play the Immediate action $S$ strategy with probability $\beta$ and player the Randomised Stopping Time Strategy with probability $1-\beta$.
\end{proof}

\section{Proof of \Cref{lemma:comparisonoftwoupperbounds} and \Cref{prop:howcutoffschangeswithcost}}

By definition, 
$
\frac{p^{M}}{1-p^{M}}
:=
\frac{-\left(u_{L}-\underline{\triangle}_{L}\right)}{u_{H}-\underline{\triangle}_{H}}
$
and
$
\frac{\bar{p}}{1-\bar{p}}
=
\frac{-u_{L}-\frac{c}{b}}{\frac{c}{b}}.
$
When 
$
\frac{c}{b}
<
\frac{-u_{L}\left(u_{H}-\underline{\triangle}_{H}\right)}
{u_{H}-\underline{\triangle}_{H}-\left(u_{L}-\underline{\triangle}_{L}\right)},
$
we have
$
p^{M}<\bar{p}.
$
\Cref{lemma:comparisonoftwoupperbounds} is shown.

The cutoff $\underline{p}$ satisfies
\begin{align}
\label{equ:p_lowerbardefinition2}
\left[u_{H}-\frac{c}{a}\right]\frac{\underline{p}}{1-\underline{p}}
+
K
\left(\frac{\underline{p}}{1-\underline{p}}\right)^{\frac{b}{b-a}}
=
\frac{c}{b}
\end{align}
where
$
K
=
\frac{c}{b}\left(\frac{b}{a}-1\right)
\left(\frac{\bar{p}}{1-\bar{p}}\right)^{1-\frac{b}{b-a}}.
$
When $c \rightarrow 0$, $\bar{p} \rightarrow 1$ and $\underline{p} \rightarrow 0$.

When $c \rightarrow 0$,
the cutoff
$
\frac{\tilde{p}}{1-\tilde{p}}
=\frac{-bu_{L}-c}{a\bar{\triangle}_{H}+c}
\rightarrow
\frac{-bu_{L}}{a\bar{\triangle}_{H}}
$
and hence 
$
\tilde{p}
\rightarrow
\frac{\frac{b}{a}\frac{-u_{L}}{\bar{\triangle}_{H}}}{1+\frac{b}{a}\frac{-u_{L}}{\bar{\triangle}_{H}}}
\in
\left(0,1\right).
$

\section{Proof of \Cref{prop:competition}}
\label{appen:competitionproof}

Suppose player $j$ uses the pure strategy
$
\left(
T^{PS}, S
\right)
$
where
$
T^{PS}
=
\frac{1}{a}
\log
\frac{\bar{\triangle}_{H}}{-\left(u_{H}-\bar{\triangle}_{H}\right)}.
$
I show that when the conditions in \Cref{prop:competition} are satisfied, player $i$'s best respond is to use the same pure strategy.
The method is `guess and verify'.
Given that player $j$ uses the pure strategy 
$
\left(
T^{PS}, S
\right)
$,
player $i$'s time $t$ payoff from taking $R$ is
\begin{align*}
U_{R}\left(t\right)
:=\begin{cases}
p_{t}\left[u_{H}-\left(1-e^{-at}\right)\bar{\triangle}_{H}\right]+\left(1-p_{t}\right)u_{L} 
&
t<T^{PS}
\\
p_{t}\left[u_{H}-\left(1-e^{-aT^{PS}}\right)\bar{\triangle}_{H}\right]+\left(1-p_{t}\right)u_{L} 
&
t\geq T^{PS}
\end{cases}
.
\end{align*}
Let 
$$
U
\left(
t
\right)
:=
\max
\left\{
U_{R}
\left(
t
\right),
0
\right\}
$$
be player $i$'s time $t$ payoff if she stops and takes an irreversible action.
Let $W\left(\cdot\right)$ be the value function associated with player $i$'s best response.
Then, the value function satisfies the following HJB equation
\begin{align}
\label{eq:HJBcompetition}
\max
\left\{ 
p_{t}a
\left[
\max
\left\{ 
u_{H}-\left(1-e^{-at}\right)\bar{\triangle}_{H}
,
0
\right\} 
-
W\left(t\right)
\right]
-
\left(1-p_{t}\right)b
W\left(t\right)
+W'\left(t\right)-c,
\right.
\nonumber
\\
\left.
U\left(t\right)-W\left(t\right)
\right\}  
& =0
.
\end{align}
Let
\begin{align*}
\frac{p^{NR}}{1-p^{NR}}
:=
\frac{\frac{c}{b}}{\frac{1}{2}\left(2u_{H}+\bar{\triangle}_{H}\right)-\frac{c}{a}-u^{S}+\psi\left(T^{PS}\right)}
\end{align*}
where
\begin{align*}
\psi\left(t\right)
:=
\left[\frac{1}{2}e^{-at}\bar{\triangle}_{H}+\frac{c}{a}-\left(-u_{L}-\frac{c}{b}\right)\frac{1-p_{t}}{p_{t}}\right]e^{-at}
\end{align*}
I prove the following lemma.

\begin{lemma}
When $c$ sufficiently small,
if 
$
p_{0}
\in
\left(
p^{NR},\tilde{p}
\right)
$,
the value function 
$
W\left(\cdot\right)
$
is
\begin{align*}
W\left(t\right)
=
\begin{cases}
W_{L}\left(t\right) 
&
t<T^{PS}
\\
0
&
t\geq T^{PS}
\end{cases}
\end{align*}
where
\begin{align*}
W_{L}\left(t\right)
=
p_{t}
\left(
\left(u_{H}^{R}-\bar{\triangle}_{H}\right)
-\frac{c}{a}+\frac{1}{2}\bar{\triangle}_{H}e^{-at}
\right)
-
\left(1-p_{t}\right)
\frac{c}{b}+p_{t}e^{at}
\psi
\left(T^{PS}\right).
\end{align*}
\end{lemma}

\begin{proof}
Let 
\begin{align*}
H\left(t,W\left(t\right),W'\left(t\right)\right)
:=
&
p_{t}a\left[\max\left\{ u_{H}-\left(1-e^{-at}\right)\bar{\triangle}_{H},0\right\} -W\left(t\right)\right]
\\
&-
\left(1-p_{t}\right)bW\left(t\right)+W'\left(t\right)-c.
\end{align*}
I show that 
$
W\left(t\right)
$
is a viscosity solution of the HJB equation (\ref{eq:HJBcompetition}).
For all the points where 
$
W\left(t\right)
$
is differentiable:
if $t<T^{PS}$,
then $W\left(t\right)\geq U\left(t\right)$;
if $t>T^{PS}$,
then 
$H\left(t,W\left(t\right),W'\left(t\right)\right)\leq0$.
At the point where $W\left(t\right)$ is not differentiable, that is, $t=T^{PS}$, I show that
$H\left(t,W\left(t\right),z\right)\geq0$
for $z\in D^{+}$
where $D^{+}=\left[W_{L}'\left(T^{PS}\right),0\right]$
and 
$H\left(t,W\left(t\right),z\right)\leq0$
for $z\in D^{-}$
where $D^{-}=\emptyset$.

\paragraph{Step 1}
For $t<T^{PS}$,
\begin{align*}
W\left(t\right)-U\left(t\right)
&=
W_{L}\left(t\right)-U_{R}\left(t\right)
\\
&=
p_{t}\left[-\frac{1}{2}\left(\bar{\triangle}_{H}\right)e^{-at}+e^{at}\psi\left(T^{PS}\right)-\frac{c}{a}\right]+\left(1-p_{t}\right)\left(-u_{L}-\frac{c}{b}\right).
\end{align*}
When 
$p_{0}<\tilde{p}$, $\psi\left(t\right)$ increases in $t$.
Then,
\begin{align*}
W\left(t\right)-U\left(t\right)
&>
p_{t}\left[-\frac{1}{2}\left(\bar{\triangle}_{H}\right)e^{-at}+e^{at}\psi\left(t\right)-\frac{c}{a}\right]+\left(1-p_{t}\right)\left(-u_{L}-\frac{c}{b}\right)
=0
\end{align*}

\paragraph{Step 2}
When $t>T^{PS}$, 
$
W\left(t\right)=0
$ 
and
$
W'\left(t\right)=0.
$
Therefore,
$
H\left(t,W\left(t\right),W'\left(t\right)\right)
=
-c
<
0.
$

\paragraph{Step 3}
I first show that for sufficiently small $c$,
$
W_{L}'\left(T^{PS}\right)
<
0.
$
We have
\begin{align*}
W_{L}'\left(t\right)
=
&
\frac{dp_{t}}{dt}\left[u_{H}-\bar{\triangle}_{H}-\frac{c}{a}+\frac{1}{2}\bar{\triangle}_{H}e^{-at}\right]
-
p_{t}\frac{1}{2}a\bar{\triangle}_{H}e^{-at}
\\
&
+\frac{dp_{t}}{dt}\frac{c}{b}+\frac{dp_{t}}{dt}e^{at}\psi\left(T^{PS}\right)+p_{t}ae^{at}\psi\left(T^{PS}\right)
\end{align*}
and
if
$
\frac{c}{a}
<
\left(-u_{L}-\frac{c}{b}\right)\frac{1-p_{T^{PS}}}{p_{T^{PS}}}
$,
\begin{align*}
W_{L}'\left(T^{PS}\right)
=
&
\frac{dp_{t}}{dt}\mid_{t=T^{PS}}\left[\frac{c}{b}-\left(-u_{L}-\frac{c}{b}\right)\frac{1-p_{T^{PS}}}{p_{T^{PS}}}\right]+p_{T^{PS}}a\left[\frac{c}{a}-\left(u^{S}-u_{L}-\frac{c}{b}\right)\frac{1-p_{T^{PS}}}{p_{T^{PS}}}\right]
\\
&<0.
\end{align*}
When 
$
t=T^{PS}
$, 
$
H\left(t,W\left(t\right),W'\left(t\right)\right)=0
$.
Since 
$
H\left(t,W\left(t\right),z\right)
$
increases in 
$z$,
we have 
$
H\left(t,W\left(t\right),z\right)\geq0
$ 
for 
$
z\in D^{+}
$.

\end{proof}

\section{Multi-player extension: proof of \Cref{prop:N-playerequilibrium}}
\label{appen:multi-playerextension}

In this extension, I generalise the two-player model to a multi-player model.
There are $N>2$ players in this model.
The first $R$ taker gets the first prize $u_{\omega}$ in state $\omega \in \left\{H,L\right\}$ and all other $R$ takers get the second prize $u_{\omega}-\bar{\triangle}_{\omega}$.
In case of the simultaneous move, the payoff is assumed to be the convex combination of the first and second prizes.
I focus on discussing the conditions for the existence of the learning equilibrium where the players use the random stopping strategy.

Suppose the players use the mixed strategy 
$
\left(\rho,\sigma\right)
$
defined as in \Cref{def:mixedstrategy},
where $\rho$ is the probability that the player stops and takes $R$ before or at time $t$ conditional on no revealing signal.
By definition, $\rho$ must be weakly increasing.
I derive a necessary condition for
$
\frac{d\rho}{dt}\geq0.
$
Let 
$
Q_{\omega}\left(t\right)
$
be the probability that no one has taken $R$ before or at time $t$ in state $\omega$.
Since I focus on symmetric equilibrium, we have
\begin{align}
\label{equ:Qequation}
Q_{\omega}\left(t\right)
=
\left(
1-
F_{\omega}\left(t\right)
\right)
^{N-1}
\end{align}
where
$
F_{\omega}\left(t\right)
$
is the probability that a player takes $R$ before or at time $t$ in state $\omega$.
We have
\begin{align}
\label{equ:FHequation}
F_{H} \left(t\right) 
=
1-e^{-at}\left(1-\rho  \left(t\right)\right)
\end{align}
and
\begin{align}
\label{equ:FLequation}
F_{L} \left(t\right) 
=\intop_{0}^{t}\left[e^{-bs}\frac{d\rho  \left(s\right)}{ds}\right]ds.
\end{align}
Let
$$
U_{R}\left(t\right)
:=
p_{t}\left[u_{H}-\left(1-Q_{H}\left(t\right)\right)\bar{\triangle}_{H}\right]
+
\left(1-p_{t}\right)\left[u_{L}-\left(1-Q_{L}\left(t\right)\right)\bar{\triangle}_{L}\right]
$$
be the time $t$ payoff from taking $R$.
When the players randomise between taking $R$ and acquiring information for $dt$ longer, the indifference condition in equilibrium is
$$
p_{t}adt\left[u_{H}-\left(1-Q_{H}\left(t+dt\right)\right)\bar{\triangle}_{H}\right]
+
\left(1-p_{t}adt-\left(1-p_{t}\right)bdt\right)U_{R}\left(t+dt\right)
-
cdt
=
U_{R}\left(t\right).
$$
When $dt\rightarrow 0$, we have the following indifference condition
\begin{align}
\label{eq:indifferenceconditionwrtQ}
c
=
\left(1-p_{t}\right)b\left[-\left(u_{L}-\left(1-Q_{L}\left(t\right)\right)\bar{\triangle}_{L}\right)\right]
+
\left(1-p_{t}\right)\frac{dQ_{L}}{dt}\left(t\right)\bar{\triangle}_{L}
+
p_{t}\frac{dQ_{H}}{dt}\left(t\right)\bar{\triangle}_{H}.
\end{align}
Given (\ref{equ:Qequation}), (\ref{equ:FHequation}) and (\ref{equ:FLequation}),
(\ref{eq:indifferenceconditionwrtQ}) can be written as a differential equation that involves $\rho$ and its derivatives.
That is,
\begin{align*}
\frac{d\rho}{dt}\left(t\right)
=
\frac{b\left(-u_{L}\right)-c-L_{t}c-b\bar{\triangle}_{L}\left(1-F_{L}\left(t\right)\right)^{N-1}-L_{t}\left(N-1\right)\bar{\triangle}_{H}\left(1-F_{H}\left(t\right)\right)^{N-2}ae^{-at}\left(1-\rho\left(t\right)\right)}{\left(N-1\right)\bar{\triangle}_{L}\left(1-F_{L}\left(t\right)\right)^{N-2}e^{-bt}+L_{t}\left(N-1\right)\bar{\triangle}_{H}\left(1-F_{H}\left(t\right)\right)^{N-2}e^{-at}}.
\end{align*}
To have 
$
\frac{d\rho}{dt}
\geq
0
$,
we need the numerator to be positive (as the denominator is positive).
A necessary condition for 
$
\frac{d\rho}{dt}
\geq
0
$
is
$$
\frac{p_{0}}{1-p_{0}}
\leq
\frac{b\left(-u_{L}\right)-c}{\left(N-1\right)\bar{\triangle}_{H}a+c}
=:
\frac{\tilde{p}_{N}}{1-\tilde{p}_{N}}.
$$
This shows \Cref{prop:N-playerequilibrium}.

\section{Observable actions: proof of \Cref{prop:observableactionresult}}

%In order to discuss the equilibrium, I first define a strategy that is similar to the random stopping strategy as in the unobservable action case.
When actions are observable, the history contains both the public component and the private component.
The public component is the action taken or not taken by the opponent and
the private component is the signal received by the player herself.
The following observation eliminates the histories that are not interesting.

\begin{observation}
After receiving an $H$-state ($L$-state, resp) revealing signal, the player takes $R$ ($S$, resp) immediately.
\end{observation}

This observation says that after receiving a revealing signal, the player does not have incentives to postpone the action.
Specifically, she does not have incentive to conceal the fact that she has learned the state even though she knows that her action is informative.
This is because there is the first-mover advantage. 
The player gains nothing from postponing an action after she has learned the state.
Given this observation, the interesting histories are the ones associated with no revealing signal.
It is thus sufficient to check the player's strategy conditional on no arrival of a revealing signal.

Suppose the opponent uses MRSS.
I first consider the history after observing the opponent taking $S$.
Given the opponent's strategy, she only stops acquiring information and takes $S$ after observing the $L$-state revealing signal.
Therefore, after observing the opponent taking $S$, the player infers that the state is $L$. 
Her best response is hence to take $S$ immediately.

Then consider the history after observing no action taken.
First notice that the player's belief evolve in a different way.
Given the opponent's strategy, not observing any action taken indicates that the opponent has not received any revealing signal. 
As a result, the belief after observing no revealing signal and no action taken up to time $t$ is
\begin{align*}
\frac{p_{t}}{1-p_{t}}
=
e^{2\left(b-a\right)t}\frac{p_{0}}{1-p_{0}}.
\end{align*}
Let
\begin{align*}
\bar{U}_{R}\left(t\right)
:=
p_{t}u_{H}+\left(1-p_{t}\right)u_{L}
\end{align*}
be the player's time $t$ payoff from taking $R$ if she is the first $R$ taker and let
\begin{align*}
\underline{U}_{R}\left(t\right)
:=
p_{t}\left(u_{H}-\bar{\triangle}_{H}\right)+\left(1-p_{t}\right)\left(u_{L}-\bar{\triangle}_{L}\right)
\end{align*}
be the player's time $t$ payoff from taking $R$ if she is the second $R$ taker.
Since 
$
p_{0}>p^{L}
$, 
we have 
$
p_{t}
>
p^{L}
$ 
and hence 
$
\bar{U}_{R}\left(t\right)
>
0
$.
The player is indifferent between taking $R$ and acquiring information for $dt$ longer if
\begin{align*}
& 
p_{t}\left\{ adt\left[u_{H}+\left(1-\left(1-adt\right)\left(1-h\left(t\right)dt\right)\right)\bar{\triangle}_{H}\right]\right.
\\
& 
\left.+\left(1-adt\right)\left[\left(1-adt\right)\left(1-h\left(t\right)dt\right)\bar{U}_{R}\left(t+dt\right)+\left(1-\left(1-adt\right)\left(1-h\left(t\right)dt\right)\right)\underline{U}_{R}\left(t+dt\right)\right]\right\} 
\\
& +\left(1-p_{t}\right)\left\{ \left(1-bdt\right)\left[\left(bdt+\left(1-bdt\right)\left(1-h\left(t\right)dt\right)\right)\bar{U}_{R}\left(t+dt\right)+\left(1-bdt\right)h\left(t\right)dt\underline{U}_{R}\left(t+dt\right)\right]\right\} 
\\
&
-cdt
\\
&
=
\\
& \bar{U}_{R}\left(t\right).
\end{align*}
When $dt \rightarrow 0$,
the equation above is equivalent to
\begin{align*}
& 
p_{t}a\left(u_{H}-\bar{U}_{R}\left(t\right)\right)+\left(1-p_{t}\right)b\left(-\bar{U}_{R}\left(t\right)\right)-p_{t}a\left(p_{t}\bar{\triangle}_{H}+\left(1-p_{t}\right)\bar{\triangle}_{L}\right)
\\
&
-h\left(t\right)\left(p_{t}\bar{\triangle}_{H}+\left(1-p_{t}\right)\bar{\triangle}_{L}\right)+\frac{d\bar{U}_{R}\left(t\right)}{dt}
=
c.
\end{align*}
The hazard rate
\begin{align*}
h\left(t\right)
=
\frac{b\left(-u_{L}\right)-c-\frac{p_{t}}{1-p_{t}}\left[a\left(p_{t}\left(\bar{\triangle}_{H}-\bar{\triangle}_{L}\right)+\bar{\triangle}_{L}\right)+c\right]}{\frac{p_{t}}{1-p_{t}}\bar{\triangle}_{H}+\bar{\triangle}_{L}}
\end{align*}
is positive if and only if
\begin{align*}
\frac{p_{t}}{1-p_{t}}
<
\frac{b\left(-u_{L}\right)-c}{a\left[p_{t}\left(\bar{\triangle}_{H}-\bar{\triangle}_{L}\right)+\bar{\triangle}_{L}\right]+c}.
\end{align*}
Since \Cref{prop:observableactionresult} assumes
$
\bar{\triangle}_{H}
=
\bar{\triangle}_{L}
$
and
$
p_{0}
<
\tilde{p},
$
the hazard rate
$
h\left(t\right)
>
0.
$

Last,
consider the history after the opponent taking $R$.
Let $p_{t}^{R}$ be the player's belief after observing the opponent taking $R$ at time $t$.
Since the opponent takes $R$ with positive rate after no revealing signal, the player's belief $p_{t}^{R}$ is smaller than one.
After the opponent has taken $R$, the player becomes the only player in this game. 
The payoff associated with $R$ now is
$
u_{\omega}-\bar{\triangle}_{\omega}
$
in state $\omega \in \{H,L\}$.
The player's best response is to use the single DM optimal strategy.

\end{document}